\documentclass{article}

\usepackage[final,nonatbib]{neurips_2022}
\usepackage{amsmath}

\usepackage{amsthm}
\usepackage{subcaption}
\usepackage{graphicx}
\usepackage[toc,page]{appendix}
\usepackage{lipsum,float}
\usepackage[font=small,labelfont=bf]{caption}
\newtheorem{theorem}{Theorem}[section]

\newtheorem{lemma}[theorem]{Lemma}

\usepackage[most]{tcolorbox}

\usepackage{tikz}
\usepackage{tikz-cd}
\usepackage{makecell}

\newtheorem{proposition}{Proposition}
\DeclareMathOperator{\sgn}{sgn}

\usepackage{array}

\usepackage[utf8]{inputenc} 
\usepackage[T1]{fontenc}    
\usepackage{hyperref}       
\usepackage{url}            
\usepackage{booktabs}       
\usepackage{amsfonts}       
\usepackage{nicefrac}       
\usepackage{microtype}      
\usepackage{xcolor}         
\usepackage{epigraph}
\usepackage{bm}
\usepackage{arydshln}
\allowdisplaybreaks

\usepackage[style=numeric-comp]{biblatex}
\title{Divide-and-Conquer Dynamics in \\ AI-Driven Disempowerment}

\author{%
 Peter S. Park \\
  Department of Physics\\
 MIT
   \And
   Max Tegmark \\
  Department of Physics\\
  MIT\\
}

\begin{document}

\maketitle
 \raggedbottom

\begin{abstract}
 AI companies are attempting to create AI systems that outperform humans at most economically valuable work. Current AI models are already automating away the livelihoods of some artists, actors, and writers. But there is infighting between those who prioritize current harms and future harms. We construct a game-theoretic model of conflict to study the causes and consequences of this disunity. Our model also helps explain why throughout history, stakeholders sharing a common threat have found it advantageous to unite against it, and why the common threat has in turn found it advantageous to divide and conquer. 
 
 Under realistic parameter assumptions, our model makes several predictions that find preliminary corroboration in the historical-empirical record. First, current victims of AI-driven disempowerment need the future victims to realize that their interests are also under serious and imminent threat, so that future victims are incentivized to support current victims in solidarity. Second, the movement against AI-driven disempowerment can become more united, and thereby more likely to prevail, if members believe that their efforts will be successful as opposed to futile. Finally, the movement can better unite and prevail if its members are less myopic. Myopic members prioritize their future well-being less than their present well-being, and are thus disinclined to solidarily support current victims today at personal cost, even if this is necessary to counter the shared threat of AI-driven disempowerment.
 \end{abstract}

\pagebreak

\newpage

\noindent
\begin{minipage}[t]{.45\linewidth}
  \setlength{\epigraphwidth}{1\linewidth}\epigraph{First they came for the socialists, 

and I did not speak out---

     Because I was not a socialist.

\hspace{1pt}

Then they came for the trade unionists, 

and I did not speak out---

     Because I was not a trade unionist.

\hspace{1pt}

Then they came for the Jews, 

and I did not speak out---

     Because I was not a Jew.

\hspace{1pt}

Then they came for me—

and there was no one left to speak for me.}{\textit{Martin Niem\"{o}ller}}
\end{minipage}%
\begin{minipage}[t]{.55\linewidth}
  \setlength{\epigraphwidth}{0.68\linewidth}\epigraph{United we stand, divided we fall.}{\textit{Aesop}}
\end{minipage}

\renewcommand{\arraystretch}{1.5}

\section{Introduction}

As general-purpose AI models advance in capability, career paths are increasingly getting automated away. Current examples include artists, \supercite{thorpe2023chatgpt} actors, \supercite{lawler2023aihollywood} and writers.\supercite{Kelly2023} The trajectory of AI increasingly replacing human livelihoods---rather than augmenting them---seems likely to continue.\supercite{brynjolfsson2022turing} OpenAI, the creator of ChatGPT, has the stated mission of creating ``highly autonomous systems that outperform humans at most economically valuable work.” \supercite{openai2018charter} 

OpenAI's mission is consistent with economist Daron Acemoglu's prediction that ``a very large number of people will only have marginal jobs, or not very meaningful jobs.'' \supercite{ito2023ai} Philosopher Aaron James goes further, predicting that  ``in  due course at least, [AI] really might cause lasting structural unemployment on a mass scale.'' \supercite{james} In this scenario,
\begin{quote}
``...jobs are steadily automated, year after year. In the old days, for every job destroyed, a new one was eventually created, leaving total employment more or less unchanged. Now deep-learning machines, aided by clever entrepreneurs, race ahead and do the new tasks as well...

...Many people---most people, even in their prime years—simply can’t find tolerable 
work and stop looking.''
\end{quote}

If we humans were in fact thrust into an indefinite state of economic uselessness, it is reasonable to anticipate that we humans would eventually lose our bargaining power: not just over the economic aspects of society, but over its social and political aspects as well.

Despite this, most people have not yet taken meaningful steps towards preventing this AI-driven disempowerment. AI development remains largely unregulated in the United States, the country with the world's most advanced AI industry. \supercite{joseph2023legislature} While some copyright lawsuits, \supercite{bearne2023ai} contract negotiations, \supercite{lawler2023aihollywood} and consumer boycotts \supercite{frenkel2023not} have been initiated in order to protect against and dissuade the AI automation of livelihoods, these efforts have yet to see lasting successes.  Most tellingly, AI companies and their associates continue to receive large amounts of funding to increasingly automate various parts of our society's economy, and are substantively progressing towards this goal. \supercite{metinko2023aifunding}

What accounts for the lack of effective resistance to AI-driven disempowerment? We propose that this is largely due to the disunity between current victims and future victims. The current victims of AI-driven disempowerment---those whose professions are already being automated away---are presently engaged in a lonely struggle, as society seems content to view their predicament as an unfortunate but necessary casualty of technological progress. Future victims---those who are not yet feeling the direct impact of AI displacement---underestimate the severity of these threats and the speed with which they may manifest. This underestimation induces future victims to not consider AI automation to be an imminent threat to their own interests, thereby dissuading them from materially supporting current victims' efforts against AI-driven disempowerment.

This disunity between current and future victims is exacerbated by the divide-and-conquer dynamics inherent in AI development. Powerful stakeholders who are invested in the proliferation of AI use their influence to downplay its future dangers, redirect attention towards its benefits, and portray society's transition to AI as inevitable. To illustrate, consider the following podcast quote\supercite{fridman2023zuckerberg} of Mark Zuckerberg: the founder and CEO of Meta, one of the leaders of the AI race.
\begin{quote}
``I don't know if mundane is the right word, but there are concerns that already exist, about people using AI tools to do harmful things of the type that we’re already aware...

...That’s going to be a pretty big set of challenges that the companies working on this are going to need to grapple with, regardless of whether there is an existential crisis as well sometime down the road.''
\end{quote}
Zuckerberg's downplaying of future AI harms can be interpreted as reinforcing in the minds of podcast-interview listeners a false dichotomy between current and future harms. Such divide-and-conquer dynamics can serve to isolate the current victims of AI: by persuading future victims to wait-and-see, rather than solidarily supporting current victims' efforts against their shared threat.

The human tendency to wait-and-see in the face of imminent AI automation is well-exemplified by the fate of the digital artists\supercite{wong2023artists} (see Figure~\ref{fig:artistcomparison}). The AI art models of 2022 were promising, but were not yet a complete automation threat. Some  artists presciently panicked, and saw the importance of acting right away.  \supercite{deck2022ai}   But other artists---faced with the narrative of AI art's deficiencies and  the doctrine of downplaying hype about AI \supercite{mysteryAI_2022_episode4}---were convinced to wait-and-see. \supercite{hemlock2022divided} And in just one year, AI art made a large leap in aesthetic quality, to the point of indefinitely taking away many artists' income streams.\supercite{aiart}

\begin{figure}[h]
    \centering
    \tcbox[colframe=white!30!black,
           colback=white!30]{
    \begin{subfigure}[t]{0.28\textwidth}
         \begin{minipage}{\textwidth}
            \includegraphics[width=\linewidth]{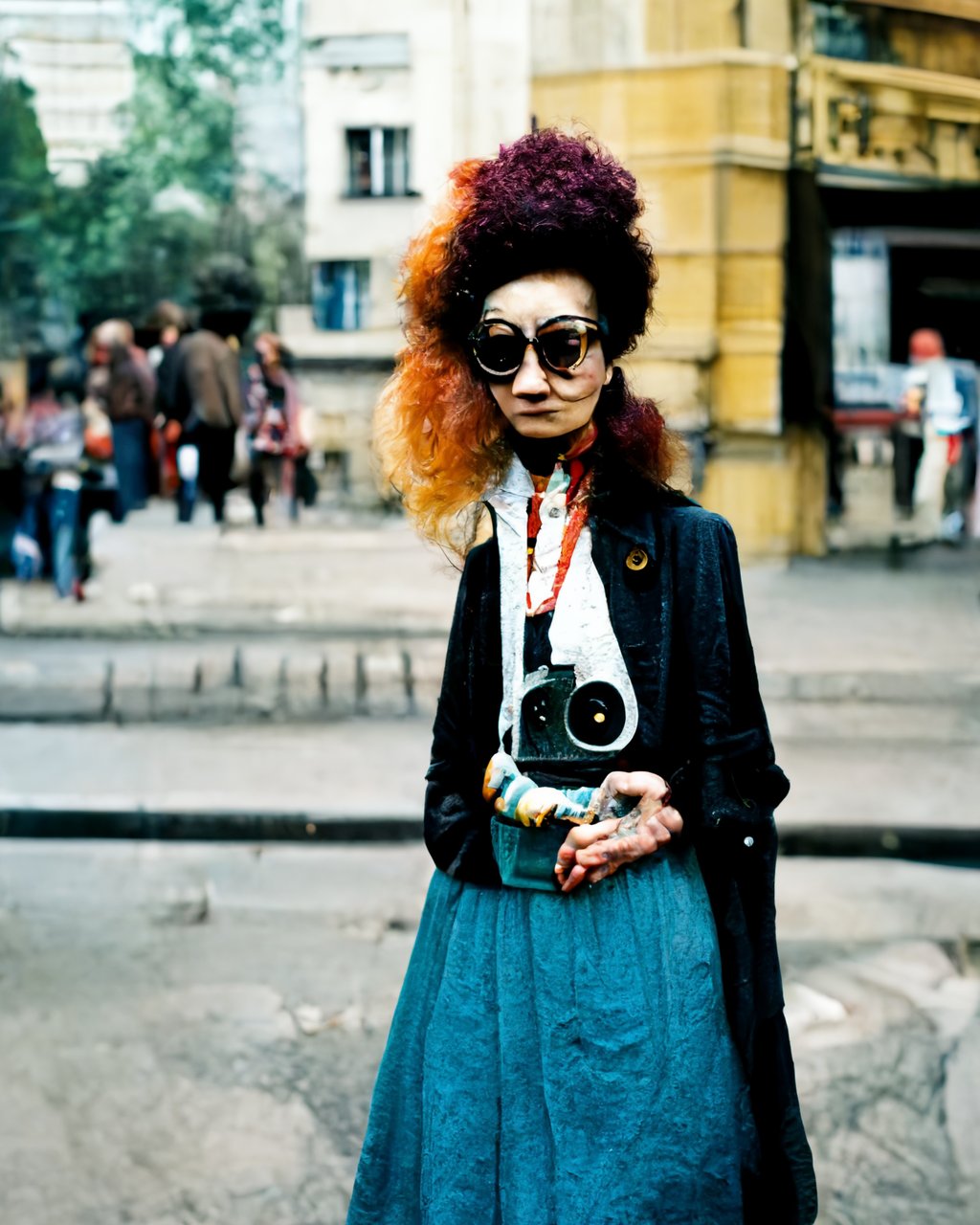}
            \caption{Midjourney\supercite{stpierre2023} \textit{(July 2022)}} 
        \end{minipage}
        \begin{minipage}{\textwidth}
           \vspace{12pt} \includegraphics[width=\linewidth]{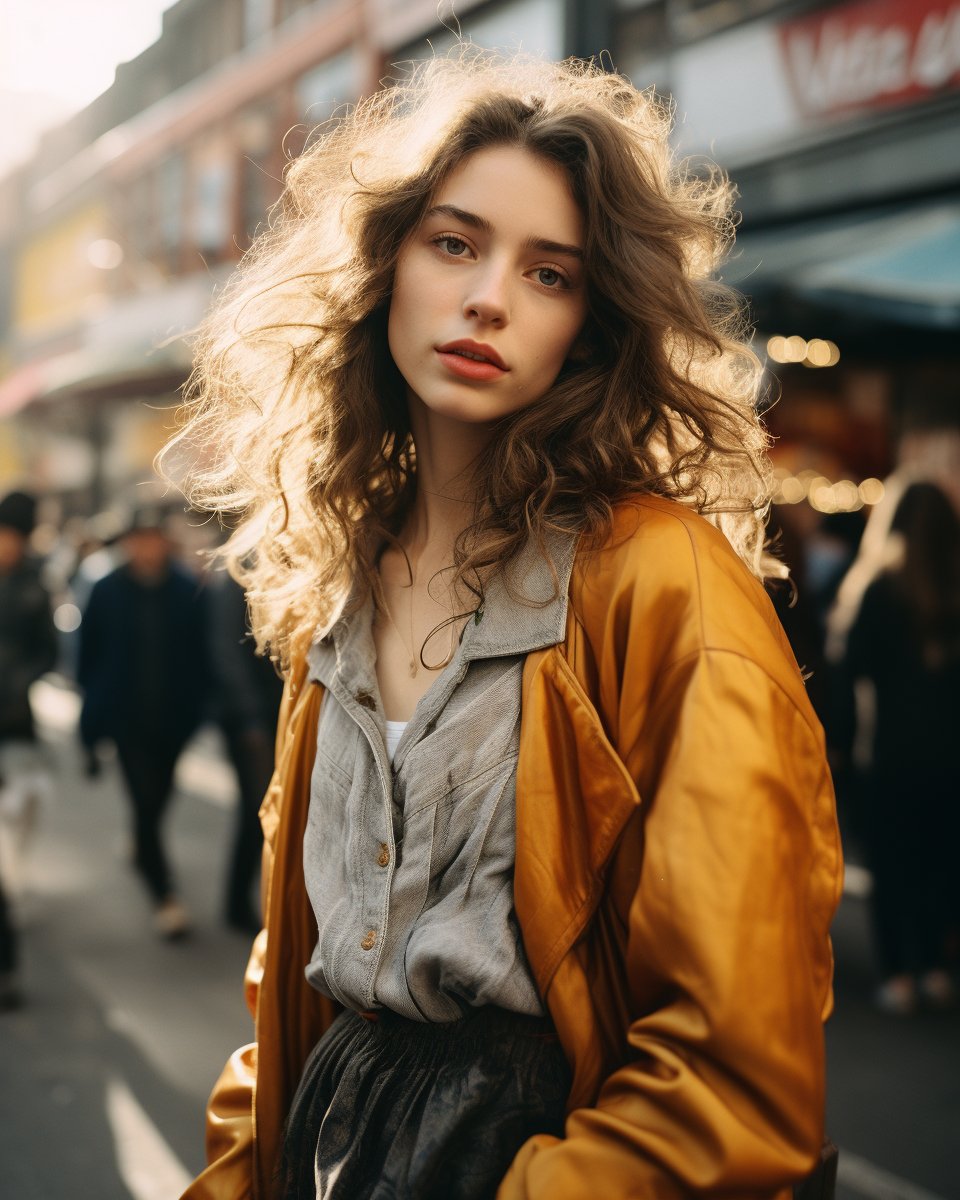}
            \caption{Midjourney\supercite{stpierre2023} \textit{(July 2023)}} 
        \end{minipage}
    \end{subfigure}
    \hspace{10pt}
   \begin{subfigure}[t]{0.6\textwidth}
            \vspace{-78pt}
        \begin{minipage}{\textwidth}
            \includegraphics[width=\linewidth]{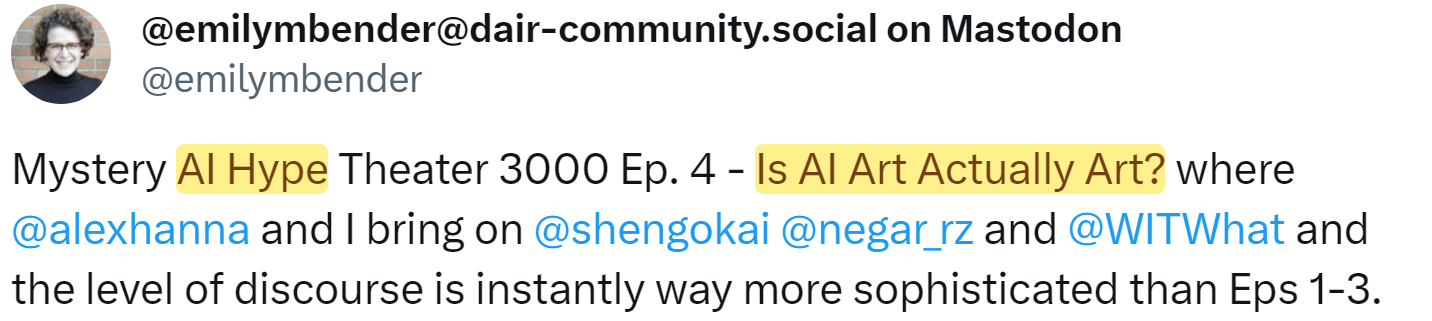}
            \caption{Emily Bender's post\supercite{bender2022mystery} \textit{(November 2022)}}
        \end{minipage}
          \begin{minipage}{\textwidth}
          \vspace{32.5pt}  \includegraphics[width=\linewidth]{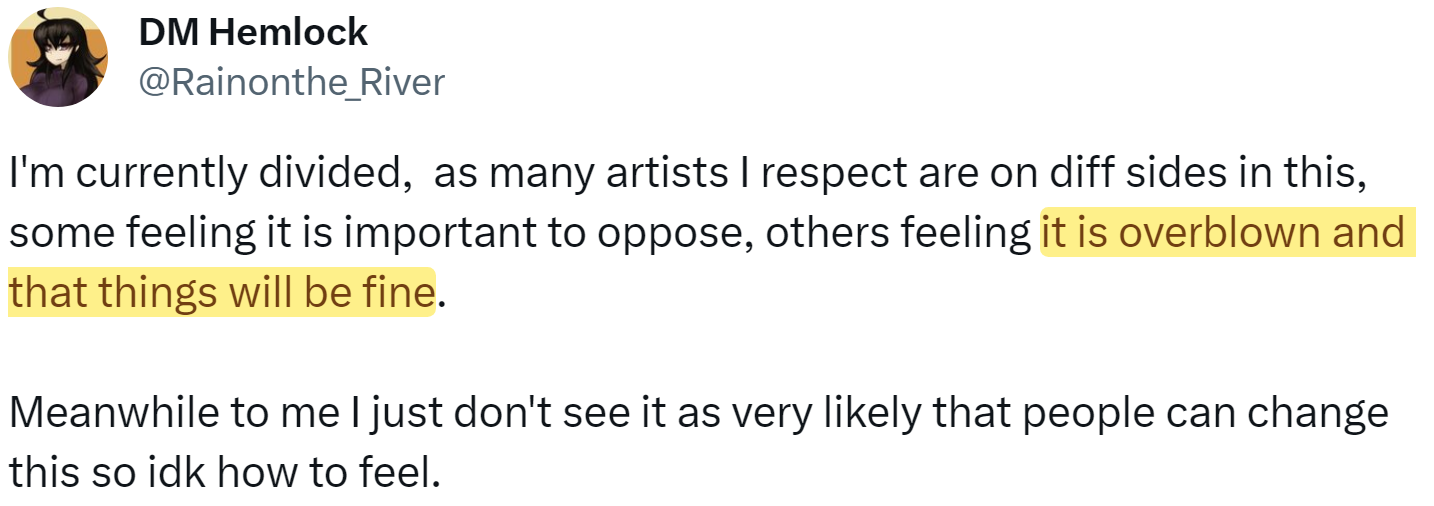}
            \caption{Artist DM Hemlock's post\supercite{hemlock2022divided} \textit{(December 2022)}} 
        \end{minipage}
          \begin{minipage}{\textwidth}
            \vspace{32.5pt}\includegraphics[width=\linewidth]{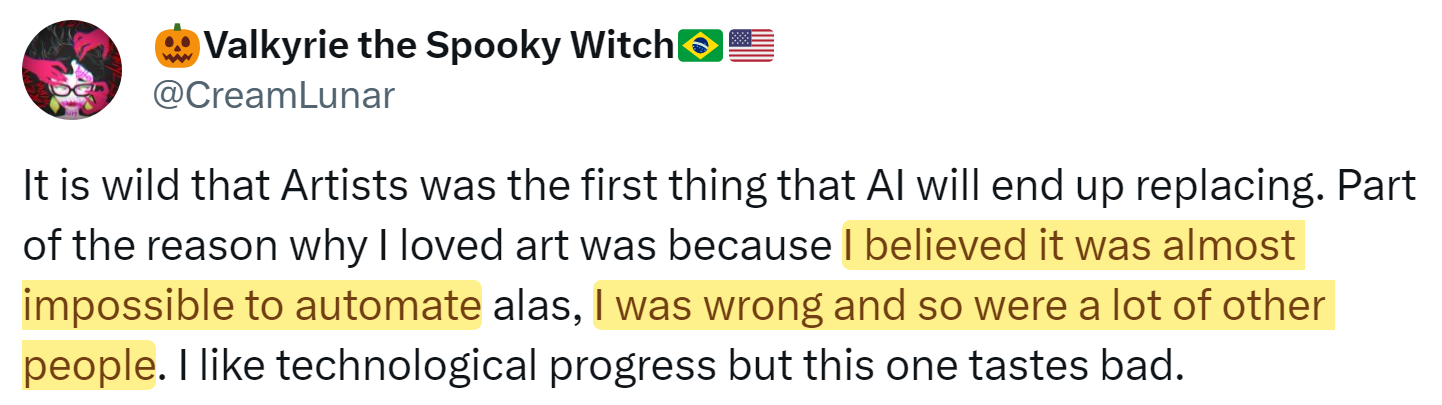}
            \caption{Artist Valkyrie the Spooky Witch's post\supercite{valkyrie2023} \textit{(January 2023)}} 
        \end{minipage}
    \end{subfigure}}
    \caption{ In 2022, many artists were convinced that the threat of their replacement by AI was overhyped, and consequently chose to wait-and-see. But in just one year, AI art made a substantial jump in aesthetic and economic competitiveness. Artists and their advocates who had not accurately panicked early enough were blindsided in 2023, likely resulting in missed opportunities.}\label{fig:artistcomparison}
\end{figure}

 Arguably, the AI art models of 2023 have already become a complete automation threat for digital artists,\supercite{obedkov2023game} especially in the event that the indefinite loss of their income streams---combined with the lack of material support from future victims, who are themselves waiting-and-seeing instead of accurately panicking---consigns the artists' copyright lawsuits and advocacy efforts to failure.

In this paper, we elucidate some causes and consequences of disunity between the current victims and future victims of AI-driven disempowerment, by proposing a game-theoretic model of conflict. The mathematical model is presented in two steps. First, the dynamical-systems part of the model is presented in Section 2. Next, the game-theoretic part of the model is presented in Section 3. In these two sections, we deduce that under realistic parameter assumptions, movement disunity and consequent inefficacy are caused by members' myopia, naivety, collaborationism, defeatism, and complacency. These five concepts are formally defined in the main text.

Then, in Section 4, we present preliminary corroboration of our predictions from the historical-empirical record. We present examples of AI-industry leaders' pronouncements that, if taken at face value, would cause people to be myopic, naive, collaborationist, defeatist, and complacent about the threat of AI-driven disempowerment. We also present historical examples of divide-and-conquer dynamics in two types of conflicts: corporate influence campaigns and military conflicts.

Finally, in Section 5, we discuss the legal, political, and societal avenues by which the conflict between the pro-human movement and the AI replacement movement is playing out---or likely will play out in the future---and suggest how our model's predictions may be productively applied to increase the unity and efficacy of the pro-human movement's efforts in each of these avenues. 

Readers who are unfamiliar with the knowhow of modeling real-world phenomena with a parsimonious formal model---and deducing from it testable, potentially generalizable predictions via theorem-proving---are encouraged to skim Sections 2 and 3 while reading the bold-font predictions, and to afterwards skip ahead to the discussion in Sections 4 and 5.

\section{The battle model}

In our dynamical-systems model, two sides are engaged in a conflict: the \emph{pro-human movement} seeking to robustly restrict AI to use cases that augment rather than replace humans (more generally, Movement 1), and the \emph{AI replacement movement}---which we also call the \emph{anti-human movement} in some figures for brevity---working to prevent such restrictions in favor of unchecked AI deployment (more generally, Movement 2). Suppose Movement 1 is comprised of $m \ge 1$ parts, and Movement 2 is comprised of $n \ge 1$ parts. The parameters $m$ and $n$ quantify the degree of disunity in each of the two movements. A movement which is divided into many parts is said to be disunited, whereas a movement with just one part is said to be completely united.

We represent the conflict with a generalized gambler's ruin model, with the state space 
\begin{equation}
I\equiv\{0, 1,\ldots, m+n\}.
\end{equation}
Each state corresponds to a battlefield, where the states are unidimensionally arranged from $i=0$ to $i=m+n$. See Figure~\ref{fig:schematic} for a schematic of this in two real-world applications of the model: the fight against AI-driven disempowerment and the Napoleonic Wars. We define the vector of battle-win probabilities (for Movement 1) by
\begin{equation}
\bm p =(p_1,\ldots, p_{m+n-1}),
\end{equation}
and assume that $0 < p_i < 1$ for all $i$.

The dynamics of the generalized gambler's ruin is as follows. The state is first set to $i=n$. Then, a battle occurs, in which Movement 1 wins with probability $p_i$ (given by a function of the model parameters, to be specified later) and Movement 2 wins with probability $1-p_i$. If Movement 1 wins, then the state is decreased by one. If Movement 2 wins, then the state is increased by one. Battles continue to occur sequentially until either the state $i=0$ (Movement 1's victory) or the state $i=m+n$ (Movement 2's victory) is achieved. 

\begin{figure}[h]
\centering

\begin{flushleft}
\textbf{(a)}
\end{flushleft}

\includegraphics[width=0.95\textwidth]{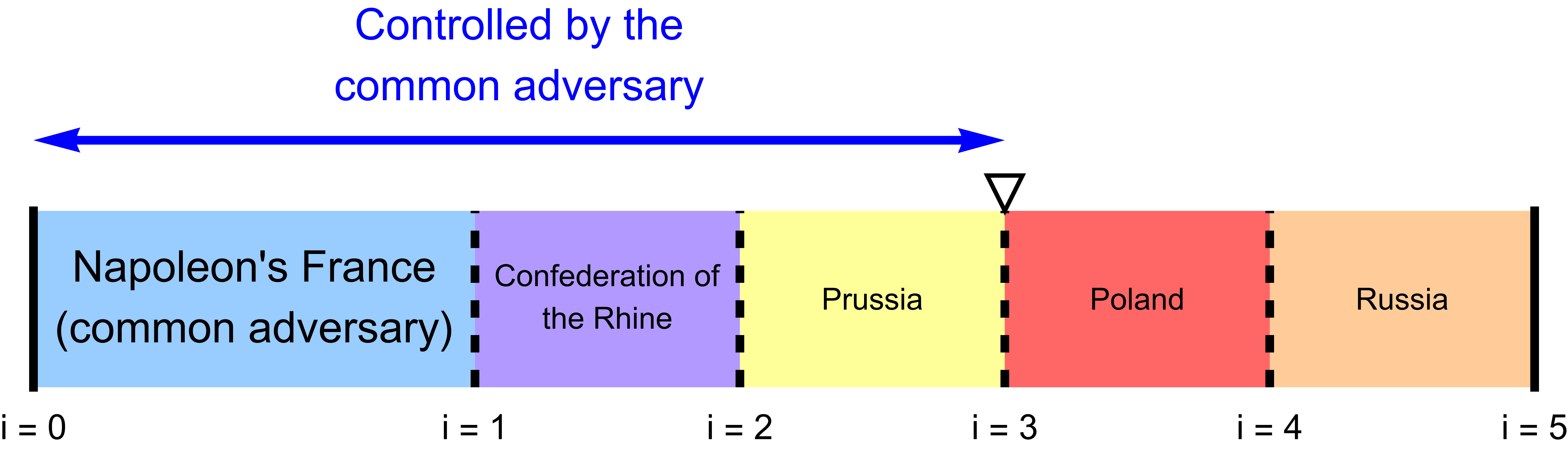}

\hspace{1pt}

\hspace{1pt}

\hspace{1pt}

\begin{flushleft}
\textbf{(b)}
\end{flushleft}

\includegraphics[width=0.95\textwidth]{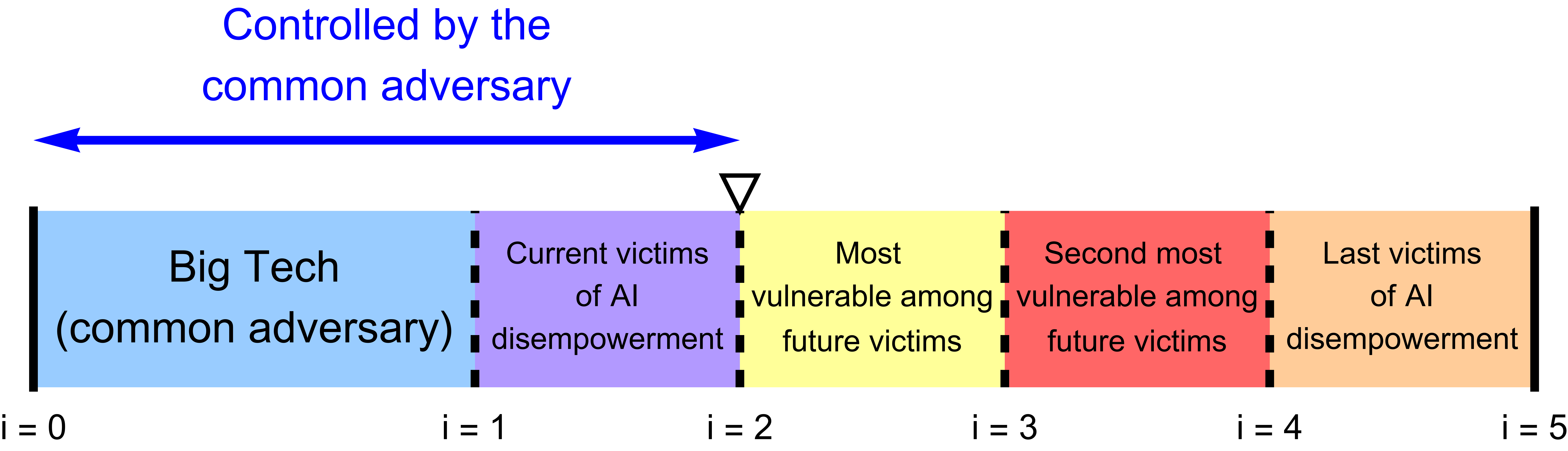}

\hspace{1pt}

\caption{\textbf{(a)}  Schematic of the conflict between Napoleonic France and the coalition of powers in opposition, as represented by our generalized gambler's ruin model. The current battlefield is $i=3$, that between Prussia and Poland. \textbf{(b)}  Schematic of the conflict between Big Tech and the members of the pro-human movement, as similarly represented by our generalized gambler's ruin model. 
The current battlefield is $i=2$, that between `current victims of AI disempowerment' and `most vulnerable among future victims.'}\label{fig:schematic}

\end{figure}

This game thus corresponds to a Markov chain with transition probability matrix
\begin{equation}
\bm{P} = \begin{pmatrix}
1 & p_1 & 0 & \cdots & 0&0&0 \\
0 & 0 & p_2 &\cdots & 0&0&0 \\
0 & 1-p_1 & 0 &\cdots & 0&0&0 \\
0 & 0 & 1-p_2 &\cdots & 0&0&0 \\
\vdots  & \vdots  & \vdots & \ddots  & \vdots& \vdots & \vdots   \\
0 & 0 & 0 & \cdots & p_{m+n-2} &0&0 \\
0 & 0 & 0 & \cdots & 0&p_{m+n-1}&0 \\
0 & 0 & 0 &\cdots & 1-p_{m+n-2}&0&0\\
0 & 0 & 0 & \cdots &0&1-p_{m+n-1} &1
\end{pmatrix}.
\end{equation}

The initial probability distribution $\bm v^0$ is given by the column vector\begin{equation}\bm v^0(s)= \begin{cases}1, & \text{ if } i=n, \\
0, & \text{ if } i \neq n,
\end{cases}\end{equation}
and the probability distribution at time step $\ell$ is given by 
$\bm v^\ell =  \bm{P}^\ell\bm v^0$. With probability one, the Markov chain arrives in a finite number of time steps at one of the two absorbing states, $i=0$ or $i=m+n$. We define $q$ as the probability of arriving at $i=0$ first, which we call the  (overall) win probability of Movement 1.

\pagebreak
\newpage

\pagebreak
\newpage

\subsection{Models for the battle-win probability $p_i$}

We first consider the simple case in which Movement 1's battle-win probability, $p_i$, does not depend on the state $i$, i.e., $p_i=p$ for all $i\in I$. Then, the win probability of Movement 1, $q$, can be explicitly calculated \supercite[Section 12.2]{CharlesGrinstead2022IP} as

\begin{equation}\label{qformula}
q=\begin{cases}\displaystyle\frac{1 - \left(\frac{1 - p}{p}\right)^m}{1 - \left(\frac{1 - p}{p}\right)^{m + n}} & \text{ if } p\neq \frac{1}{2}, \\
\displaystyle {m\over m+n} & \text{ if } p=\frac{1}{2}.
\end{cases}
\end{equation}
We will explore various models for the probability values $p_i$. 
Suppose Movement 1 has a total force strength $F$, split into $n$ pieces; and Movement 2 has a total force strength $sF$, split into $m$ pieces. Here, the strength parameter $s$ represents the relative strength advantage of Movement 2, and the abstraction of ``force strength'' aggregates a variety of resources (e.g., human capital; material resources; and less tangible types of resources like engagement opportunities, public support, and morale) that help one side over another in a battle. 

We first consider the simple case where each side's battle win probability ($p$ and $1-p$) is proportional to its force strength values $F/n$ and $sF/m$, respectively. This yields the relation
\begin{equation}
\frac{p}{1-p} 
= \frac{\frac{F}{m}}{\frac{sF}{n}}
= \frac{n}{sm}.
\end{equation}
Solving for $p$ gives 
\begin{equation}\label{functionalformp}
 p_i=p=\frac{1}{1+sm/n}.
\end{equation}
Substituting (\ref{functionalformp}) into the corresponding expression for $q$ in (\ref{qformula}) now gives the overall win probability of Movement 1 as
\begin{equation}\label{qstandardformula}
q=\frac{1-\left(\frac{sm}{n}\right)^m}{1-\left(\frac{sm}{n}\right)^{m+n}}
\end{equation}
if $s\ne n/m$, and $q=m/(m+n)$ otherwise. Thus, Movement 1's overall win probability is determined by the three variables $m, n, $ and $s$. See Figure~\ref{fig:example} for examples of this.

\begin{figure}[h]
\centering 
\begin{flushleft}
\hspace{60pt}\textbf{(a)}
\end{flushleft}

\includegraphics[width=0.48\textwidth]{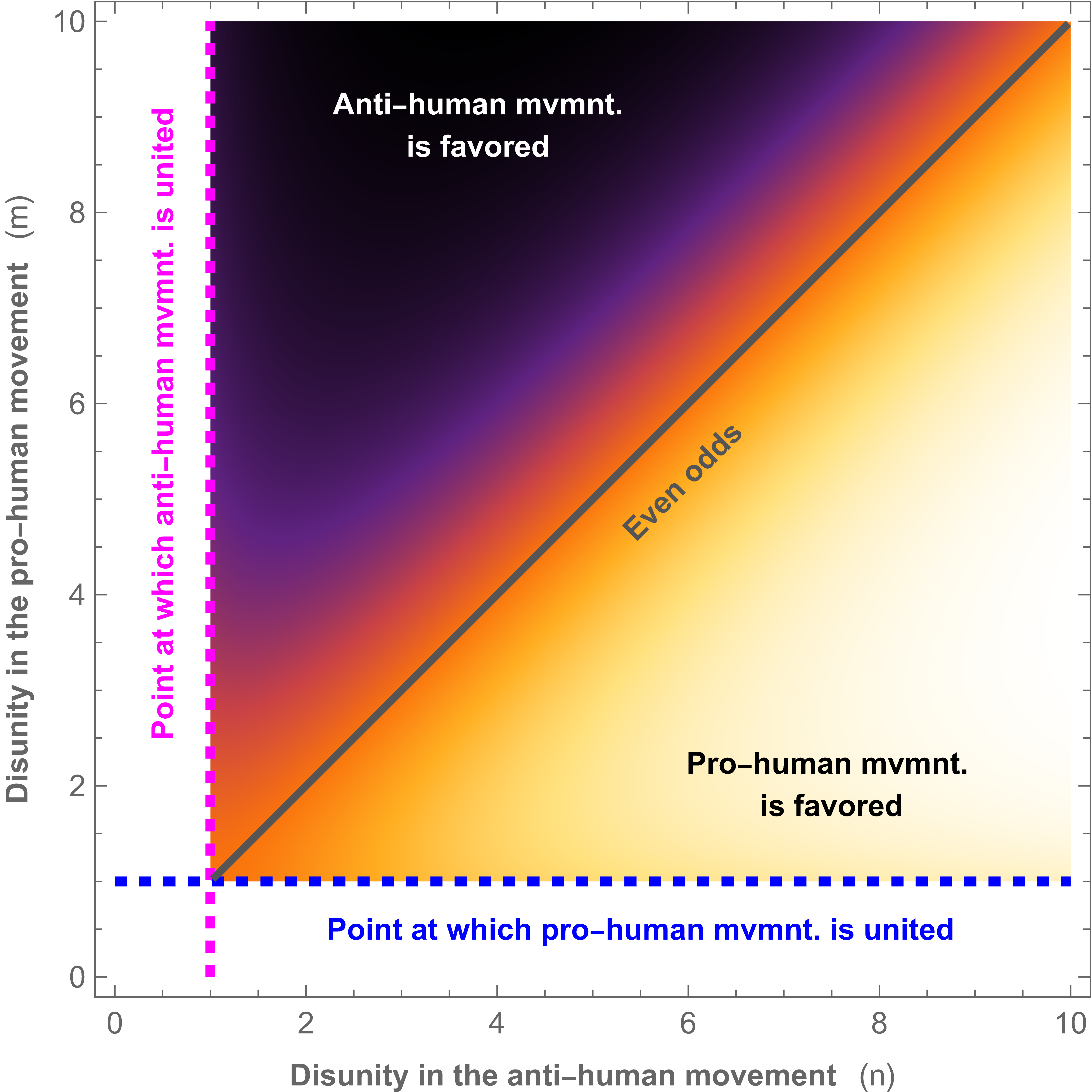} \includegraphics[width=0.14\textwidth]{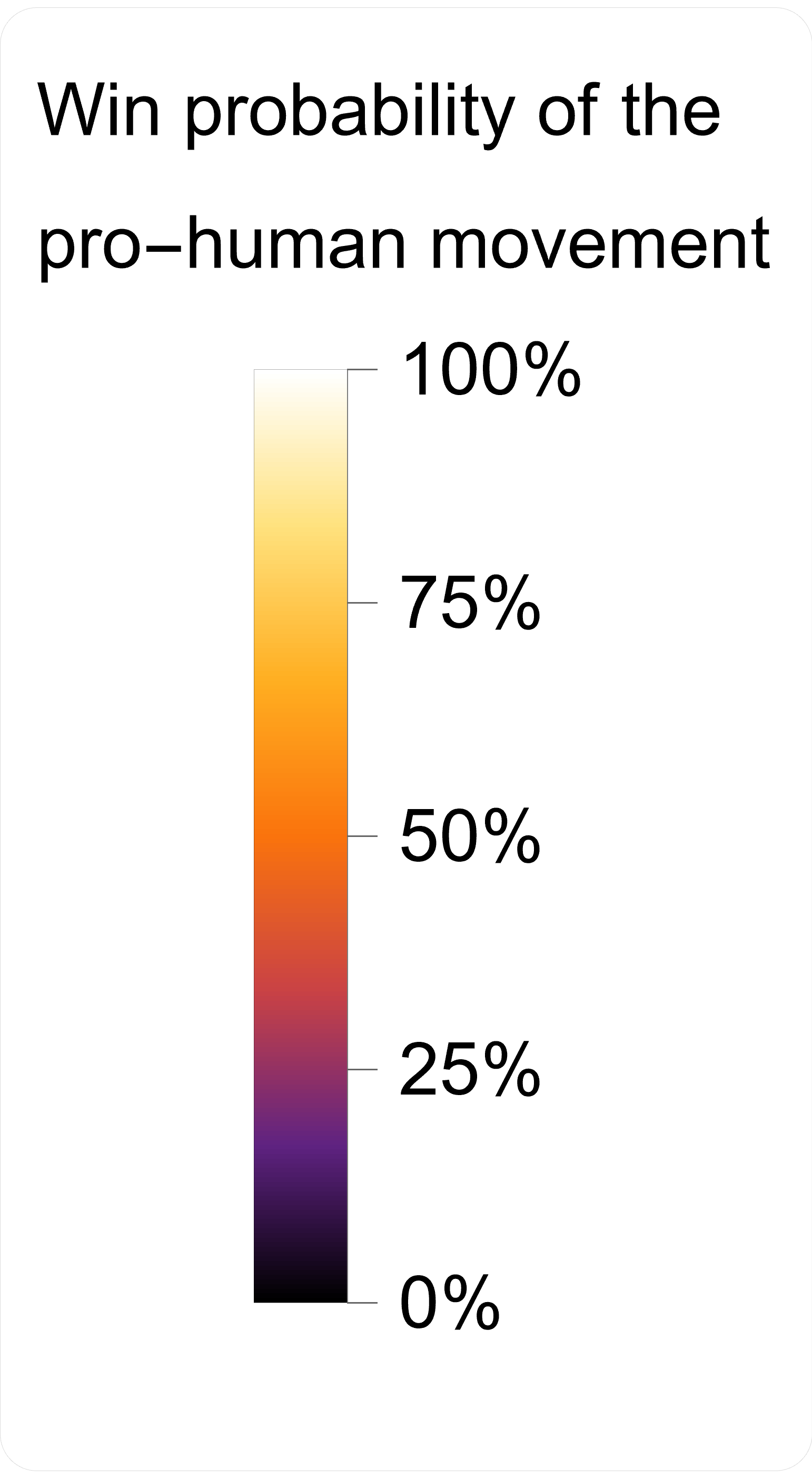}

\hspace{1pt}

\begin{flushleft}
\hspace{60pt}\textbf{(b)}
\end{flushleft}

\hspace{-16pt} \includegraphics[width=0.52\textwidth] {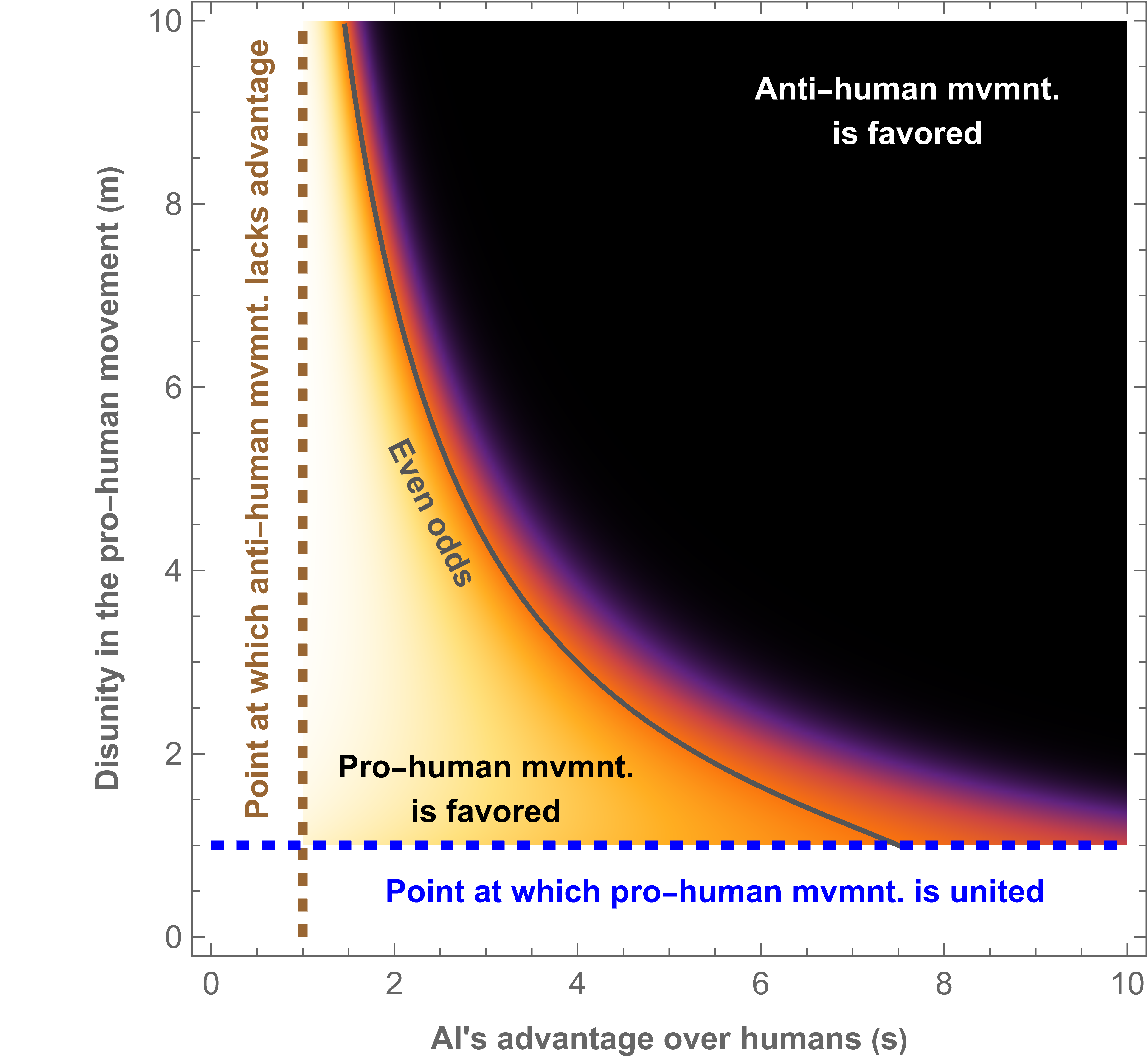}    \includegraphics[width=0.14\textwidth]{divideandconquerdynamicsfig3legend.png}    

\caption{\textbf{(a)}  Win probability of the pro-human movement as a function of the number $m$ of parts in the pro-human movement and the number $n$ of parts in the anti-human/AI replacement movement, for relative strength $s=1$. \textbf{(b)} The win probability of the pro-human movement as a function of its number of parts $m$, assuming the anti-human/AI replacement movement is comprised of $n=15$ parts. This illustrates that when $s$ is large---when the AI replacement movement is strong---full unity ($m=1$) is the optimal counterplay.
}\label{fig:example}
\end{figure}

We generalize this simple model (\ref{functionalformp}) to
\begin{equation}\label{firstgeneralform}
p = \frac{1}{1+\left(sm/n\right)^{1/R}},
\end{equation}
where $R$ denotes the \emph{randomness parameter}. When $R$ is small ($R \ll 1$), a slightly larger part has a comparatively much larger probability of winning the battle. When $R$ is large ($R\gg 1$), even a much larger part only has slightly higher than even odds of winning the battle. In particular, $R \to 0$ yields the case of deterministic outcomes, in which the stronger side always wins; and $R \to \infty$ yields the case of complete randomness, in which both sides with with probability $1/2$. Our simple model (\ref{functionalformp}) corresponds to the case $R=1$.

We further generalize our model (\ref{functionalformp}) to incorporate the spectrum between a defender's advantage and an attacker's advantage. Consider the generalization
\begin{equation}\label{generalform}
p = \sigma\left[\sigma^{-1}\left(\frac{1}{1+\left(sm/n\right)^{1/R}} \right)+\sgn(i-n) \gamma\right].
\end{equation}
Here, $\sigma$ denotes the sigmoid function 
\begin{equation}
\sigma(x) \equiv \frac{1}{1+e^{-x}},
\end{equation} which is a monotonically increasing bijection from $\mathbb{R}$ to $ (0,1)$. Its inverse function 
\begin{equation}
\sigma^{-1}(p) \equiv \log\left(p\over 1-p\right)
\end{equation}
is a monotonically increasing bijection from $(0,1)$ to $\mathbb{R}$. The parameter $\gamma$ is the \emph{attacker's/defender's-advantage parameter}. The case of $\gamma>0$ models a defender's advantage. Battles that are fought in Movement 1's territory tend to be won by the Movement $1$, and vice versa.  The case of $\gamma < 0$ models an attacker's advantage. Battles fought in Movement 1's territory tend to be won by Movement 2, and vice versa. The pre-generalized model (\ref{firstgeneralform}) corresponds to the case $\gamma=0$, where neither the defender's advantage nor the attacker's advantage exists.

\subsection{Prediction: Against a sufficiently large threat, complete unity is the optimal counterplay}


\begin{proposition}\label{prop:unitygood}
For all sufficiently large $s$, Movement $1$'s win probability 
$q$ is maximized for $m=1$ (complete unity) and monotonically decreases to zero as $m\to\infty$. 
\end{proposition}

\begin{proof}
See Section~\ref{appendix:serious} of the Appendix.
\end{proof}

In other words, it is when the threat that Movement 2 poses to Movement 1 is highly serious---when the value of $s$ is high---that complete unity becomes the strategy that maximizes Movement 1's overall win probability $q$. The qualitative interpretation is that if the AI replacement movement has a large advantage relative to the pro-human movement in battles---for example, if AI actually does have a large capacity to replace many kinds of livelihoods---then the pro-human movement may need to achieve complete unity in order to have a good shot at prevailing. See Figure~\ref{fig:threat} for an example.

\begin{figure}[H]

\centering

\begin{center}
\includegraphics[width=0.6\textwidth]{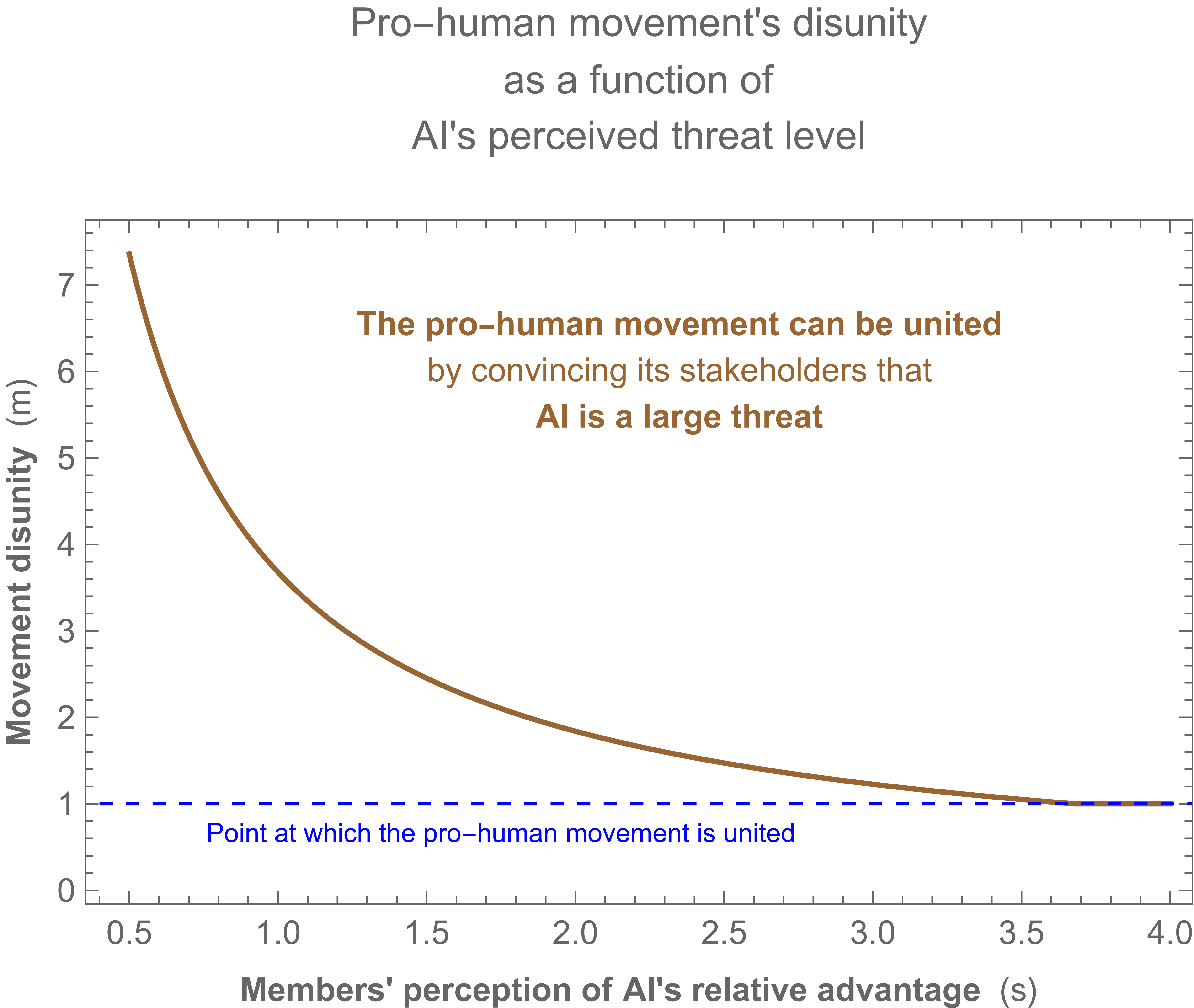}
\end{center}

\caption{ The pro-human movement's disunity as a function of members' perception of the AI replacement movement's strength parameter $s$, as $s$ is varied from $s=0.5$ to $s=4.0$ and the number of the  AI replacement movement's parts $n$ is set to $n=10$.  }\label{fig:threat}

\end{figure}

\section{Game theory: Cooperate or not?}

In the previous section, we have described our battle model. It predicts the overall win probability for each side, conditional on the two movements choosing to split into $m$ and $n$ parts, respectively. We will now describe our game-theoretic modification of this base model, which rigorously defines how each movement and its constituents decide the extent to which they cooperate, which ultimately determines their overall win probability. 

For brevity, let us write the win probability of Movement 1 from (\ref{qstandardformula}) as the function $q(m)$. In the previous section, we saw that for large enough $s$, the function 
$q(m)$ is monotonically decreasing, meaning that 
complete unity is the optimal strategy. If Movement $1$ is split into $m_0$ parts, will it therefore
benefit from unifying into $m_1$ parts, where $m_1<m_0$?
We assume that coalition-forming is voluntary, so that 
such a new coalition of $m_1$ members will only form if all 
its $M$ members agree.

We assume that there is an extra up-front cost $c> 0$ that each coalition member needs to pay if and only if they agree on this new larger coalition. We also assume that 
if Movement 1 wins, each of its members gains a periodic 
benefit $b$ for all future time periods $\ell$, which (using an exponential discounting of $e^{-r\ell}$, 
as is standard in economics) produces a total benefit of
$\sum_{\ell=1}^\infty b e^{-r\ell} = b/(e^r-1)$.
We allow for $b$ to be either positive, indicating that the member benefits from their movement prevailing; or negative, indicating that the member is in fact harmed by their movement prevailing.

In summary, the total payoff for each member of Movement 1
is given by the random variable
\begin{equation}
\Pi=\begin{cases}
\frac{b}{e^r - 1}, & \text{ if Movement $1$ adopted the status quo and won,} \\
0, & \text{ if Movement 1 adopted the status quo and lost,} \\
-c + \frac{b}{e^r - 1} & \text{ if Movement $1$ adopted greater unity and won,} \\
-c, & \text{ if Movement $1$ adopted greater unity and lost.}
\end{cases}
\end{equation}

Conditional on Movement $1$ adopting the status-quo approach $m=m_0$, its win probability is $q(m_0)$, and the expected payoff of each member of Movement 1 is thus given by
\begin{align}
\nonumber\langle\Pi(m_0)\rangle&=q(m_0)\frac{b}{e^{r}-1}  + \left(1-q(m_0)\right)\cdot 0\\&=q(m_0)\frac{b}{e^{r}-1}.\label{spayoff}
\end{align}
Conditional on Movement $1$ adopting the greater-unity approach $m=m_1$, its win probability is $q(m_1)$, and the expected payoff of each member is 
\begin{align}
\langle\Pi(m_1)\rangle &=  q(m_1)\left(-c + \frac{b}{e^{r}-1}  \right) +  \left(1-q(m_1)\right)\left(-c\right)\nonumber    \\&= -c+ q(m_1) \frac{b}{e^{r}-1}.\label{gpayoff}
\end{align}

We assume that each member is rational, in that when they choose their action in the set $A=\{S, G\}$, where $S$ denotes `voting for the status quo' and $G$ denotes  `voting for greater unity,' they are assumed to maximize their expected payoff. This motivates the definition of a \textit{(pure) Nash equilibrium}, \supercite{nash1950equilibrium} a vector of the players' actions 
\begin{equation}
(a_1,\ldots, a_{M}) \in A^{M}
\end{equation}
for which no player can deviate to the other of the two actions to gain a higher expected payoff. We follow the \textit{modus operandi} of game theory by investigating the conditions under which various outcomes are predicted as Nash equilibria of the underlying model.

Embedded in the above model is the assumption that each of the $M$ players is rational, in that they maximize their expected payoff.  The assumption of rationality is standard in game theory, and often gets grouped in with the stronger assumption of rationality that the notion of payoff that each player maximizes is an accurate reflection of their utility or well-being. This stronger assumption is not always realistic, however. In certain parts of our paper, we will relax this assumption by assuming that a given player can maximize (in expectation) a notion of utility that is incorrect, because the player is mistaken about the true value of a model parameter that is crucial to accurately predicting their actual utility. In such situation, we will clearly denote the difference between the \textit{perceived} but mistaken value of the model parameter according to the given player's belief, and the \textit{actual} value of the model parameter.

Throughout the remainder of this paper, we assume the hypothesis of Proposition~\ref{prop:unitygood}: that the strength parameter $s$ of the AI replacement movement over the pro-human movement is high enough that regardless of whether the former adopts the status-quo approach or the greater unity approach, the pro-human movement's win probability is monotonically decreasing. This assumption is realistic, given the large amount of resources currently available to the AI replacement movement \supercite{metinko2023aifunding} and empirical scaling laws which suggest that AI capabilities may continue to steadily increase. \supercite{hoffmann2022training} Under the realistic assumption of a high strength parameter $s$, greater unity directly translates to a greater win probability for the pro-human movement, a corollary that is also implicitly assumed in the rest of the paper.

\subsection{Factors that determine whether members will unite in solidarity or wait-and-see}

Recall that a member of the pro-human movement (Movement $1$) is faced with the choice of whether to vote for the status-quo approach of dividing into $m_0$ parts, or for the greater-unity approach of dividing into $m_1<m_0$ parts. The greater-unity approach requires each member to pay a cost of $c$. In return, the win probability of the pro-human movement changes from $q(m_0)$ to $q(m_1)$. Under our assumption that the strength parameter $s$ is sufficiently high, we have that $q$ is monotonically decreasing in $m$, and in particular, we have $q(m_1)>q(m_0)$. 

We assume that switching to the greater-unity approach requires unanimous support, in that any member of Movement 1 can veto it. Thus, it is meaningful to investigate the question of whether the outcome of a unanimous support of the greater-unity approach, $(a_1,\ldots, a_M)=(G,\ldots,G)$, is a Nash equilibrium. Colloquially, this constitutes the study of whether any member has an incentive to defect from this outcome, by choosing instead to veto the greater-unity approach so as to not pay their share $c$ of the cost.

A given member's decision is determined by whether their expected payoff conditional on the greater-unity approach  exceeds or is less than their expected payoff conditional on the status quo approach. The given member will choose to defect from the greater-unity approach if and only if 
\begin{equation}
\langle\Pi(m_1)\rangle  < \langle\Pi(m_0)\rangle.
\end{equation}
Substituting (\ref{gpayoff}) for the left-hand side and (\ref{spayoff}) for the right-hand side, we find that the given member will choose to defect from the greater-unity approach if and only if 
\begin{equation}
\left(\frac{ 1 }{e^{r }- 1 } \right)(b/c) \left(q(m_1) - q(m_0)\right)< 1.
\end{equation}

Further substituting in the expression (\ref{qstandardformula}), we obtain that the given member will choose to defect from the greater-unity approach if and only if the following key inequality holds:
\begin{equation}\label{inequality}
\left(\frac{ 1 }{e^{r }- 1 } \right) (b/c) \left( \frac{1-\left(\frac{\hat sm_1}{n}\right)^{m_1}}{1-\left(\frac{\hat sm_1}{n}\right)^{m_1+n}}- \frac{1-\left(\frac{\hat sm_0}{n}\right)^{m_0}}{1-\left(\frac{\hat sm_0}{n}\right)^{m_0+n}}\right)< 1.
\end{equation}

We thus see that there are three factors that determine whether a given member of the pro-human movement will defect from, and thereby prevent, the greater-unity approach. The first factor, $\frac{ e^{-r } }{1- e^{-r } }$, is determined by the exponential discount parameter $r$. The second factor, $b/c$, is determined by the (perceived) cost $c$ paid by the member if the pro-human movement opts for the greater-unity approach relative to the status-quo approach, as well as the (perceived) benefit $b$ to the member from the pro-human movement's victory relative to the AI replacement movement's victory. The third factor, 
\begin{equation}\label{thirdfactor}
\left( \frac{1-\left(\frac{\hat sm_1}{n}\right)^{m_1}}{1-\left(\frac{\hat sm_1}{n}\right)^{m_1+n}}- \frac{1-\left(\frac{\hat sm_0}{n}\right)^{m_0}}{1-\left(\frac{\hat sm_0}{n}\right)^{m_0+n}}\right),
\end{equation}
is determined by the parameters of (\ref{qstandardformula}) defined in Section 2: $s, n,$ and the values $m_1$ and $m_0$ chosen for $m$. This completes our list of model parameters, which can be found in Table~\ref{tab:variables}.

\begin{table}[h]
\centering
\caption{Variables and the quantity in the model that each of them represents. The variables listed above the bolded line pertain to the dynamical-systems part of our model.  The variables listed below the bolded line pertain to the game-theoretic part of our model.  }\label{tab:variables}

\vspace{7pt}

\begin{tabular}{ | m{5em} | m{10cm}| } 
\hline
\textbf{Variables} & \textbf{What the variable represents}  \\
\hline
   $m$ & Number of parts in the pro-human movement \\ \hline
    $n$ & Number of parts in the AI replacement movement \\ \hline
  $s$ & The strength parameter  \\ \hline
   $R$ & The randomness parameter  \\ \hline
     $\gamma$ & The attacker's/defender's-advantage parameter \\ \hline
   \Xhline{4\arrayrulewidth}
   $ r  $ & The degree to which the given member discounts their future well-being relative to their current well-being \\ \hline
      $c$ & The (perceived) cost paid by the given member if the pro-human movement opts for the greater-unity approach, relative to the status quo approach  \\ \hline
$b$ & The (perceived) benefit to the given member from the pro-human movement's victory, relative to the AI replacement movement' victory \\ \hline
   $\hat s$ & The perceived strength parameter  \\
\hline
\end{tabular}
\end{table}

We formally analyze what real or perceived conditions make it more likely that the inequality (\ref{inequality}) holds, i.e., more likely that members of Movement $1$ find it in their best interest to wait-and-see. The wait-and-see approach constitutes the decision to refrain from contributing a personal cost today, even if it means achieving greater movement unity and thereby a higher probability of prevailing over the opposing movement in the future.

\subsection{Prediction: \textit{Myopic} members tend to wait-and-see, hurting movement unity}

Suppose that the cost value $c>0$ and the benefit value $b>0$ were fixed. Then, the given future victim would choose to defect from the greater-unity approach if and only if their discount factor $r$ were sufficiently low, or in other words, sufficiently close to $0$. Indeed, as $r$ decreases to $0$, the inequality (\ref{inequality}) becomes the contradiction $\infty < c$.
Colloquially, members who prioritize their future well-being as highly as their current well-being are less likely to defect from the greater-unity approach. The reason is that as such patient members are more likely to value the long-term benefits of pro-human victory over the short-term cost they would have to pay in order to achieve that victory. 

In contrast, as $r$ increases to $\infty$, the inequality (\ref{inequality}) converges to the true identity $0<c$. Colloquially, \textit{myopic} members---those who,  relative to their current well-being, discount future well-being at a high exponential rate $r$---are more likely to defect from the greater-unity approach. The reason is that such myopic members are more likely to prioritize the short-term personal cost over the long-term benefits to pro-human victory.

That myopia reduces members' incentive to unite with current victims---and that conversely, prioritizing future well-being as much as current well-being increases members' incentive to unite with current victims---can be seen in Figure~\ref{fig:myopianaivety}.

\begin{figure}[h]

\centering

\begin{center}
\includegraphics[width=0.7\textwidth]{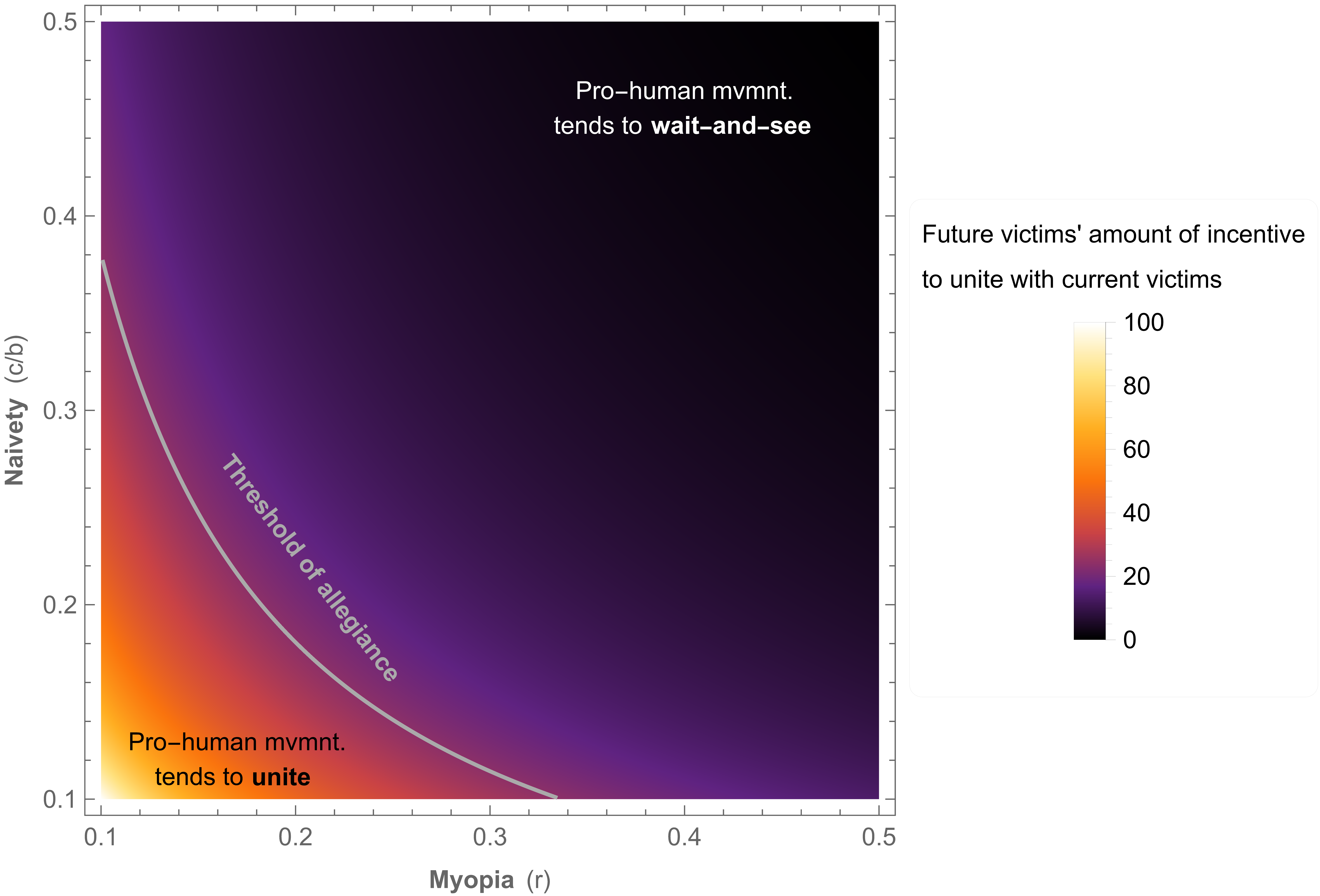}
\end{center}

\caption{The amount of incentive for uniting with current victims, as opposed to waiting-and-seeing. Members who are \emph{myopic} (deprioritize their future well-being compared to present well-being) and \emph{naive} (underestimate the future harm AI replacement poses to their own self-interest, relative to the cost of solidarily supporting current victims today) are more likely to wait-and-see. The threshold between the two preferences is defined for the parameter choices $m_0=2,m_1=1,n=50,$ and $s=10$. }\label{fig:myopianaivety}

\end{figure}

\subsection{Prediction: \textit{Naive} members tend to wait-and-see, hurting movement unity}

Suppose that the discount factor $r$ is fixed. Then, the given future victim would choose to defect from the greater-unity approach if and only if the cost it would chip in to achieve this unity, $c$, were sufficiently large in magnitude compared to the benefit $b$ they would obtain in the event of a pro-human victory. As can be seen in the inequality (\ref{inequality}), the relevant question is whether the ratio $b/c$ is larger or smaller than the threshold set by the inequality.

If the given future victim were to correctly realize that the AI replacement movement's victory would pose a substantial threat to their livelihood, then they would correctly perceive the benefit $b>0$ they would get from a pro-human victory (relative to an AI-replacement victory) to be high relative to the cost $c$. This makes it less likely that the inequality (\ref{inequality}) holds in the opposite direction, and thereby more likely that future victims would support the greater-unity approach rather than defect from it.

Suppose, on the other hand, that the given future victim is \emph{naive}, i.e., they perceive their value of $b/c$ to be small (equivalently, $c/b$ to be high) in that they underestimate the harm that the AI replacement movement posed them relative to the personal cost they would have to pay to prevent that harm. This naivety makes it more likely  that the given future victim believes the inequality (\ref{inequality}) to hold. Consequently, naivety makes it more likely that the future victim would defect from the greater-unity approach rather than support it. This is because they would avoid paying the cost required for the greater unity approach, and instead opt to wait-and-see, to the detriment of movement unity. 

That naivety reduces members' incentive to unite with current victims can be seen in Figure~\ref{fig:myopianaivety}.


\subsection{Prediction: \textit{Collaborationist} members will wait-and-see, hurting movement unity}

Consider the case that a given future victim is \textit{collaborationist}, i.e., they perceive their benefit $b$ from a pro-human victory to be negative, $b<0$. In other words, collaborationists believe that they will benefit from the victory of the AI replacement movement, not that of the pro-human movement. It follows from the definition of collaborationism, $b<0$, that the inequality (\ref{inequality}) holds unconditionally. In fact, even if the cost of the greater-unity approach for the collaborationist were zero, they would still strictly prefer the pro-human movement's loss, and consequently defect to reduce its unity and efficacy. This collaborationist approach can be explained as the maximization of the future victim's perception of their expected payoff, even if their actual expected payoff function was such that sabotaging their own side was in fact unfavorable.

\subsection{Prediction: \textit{Defeatist} members tend to wait-and-see, hurting movement unity}

Suppose that the cost value $c>0$, the benefit value $b>0$, and the discount factor $r>0$ were fixed. However, suppose also that the given future victim perceives the increase in the pro-human movement's win probability, $\left( q(m_1)-q(m_0)\right)$, with error. The expression for the factor (\ref{thirdfactor}) contains four parameters: $s, n, m_0,$ and $m_1$. In this paper, we will consider cases in which the strength parameter $s$ is misperceived as $\hat s \neq s$. Presently, we will write the expression for $q$, (\ref{qstandardformula}), as a function $q(m,s)$ of both $m$ and $s$. 

Suppose that the strength parameter $s$ is misperceived as a very large value of $\hat s$. Observe that for any fixed $m$, we have
\begin{equation}
\lim_{\hat s \to \infty} q(m,\hat s) = \lim_{\hat s \to \infty} \frac{1-\left(\frac{\hat sm}{n}\right)^m}{1-\left(\frac{\hat sm}{n}\right)^{m+n}} = 0.
\end{equation}
It follows that as $\hat s \to \infty$, we have that $q(m_0,\hat s)-q(m_1,\hat s)$ becomes arbitrarily small, since both terms converge to $0$. Thus, for sufficiently large $\hat s$, the left-hand side of the key inequality (\ref{inequality}) becomes sufficiently small, meaning that it is in the perceived self-interest of the given member to defect. It is easy to show that this result also holds for the generalized model (\ref{generalform}) of the battle-win probability, bolstering the case for the result's robustness.

Colloquially, suppose that the given future victim is \emph{defeatist}, i.e., they perceive the strength parameter $\hat s$ to be very large. In other words, a defeatist perceives the strength of the AI replacement movement over the pro-human movement to be so large that the latter's win probability is very close to $0$, regardless of whether it uses the greater-unity approach (which imposes a personal cost to the future victim) or the status quo approach (which does not). It follows that a defeatist would then defect to the status quo approach, which at least does not impose a personal cost on them at the present time.

\subsection{Prediction: \textit{Complacent} members tend to wait-and-see, hurting movement unity}

On the other hand, suppose that the strength parameter $s$ is misperceived as a very small value of $\hat s$. Observe that for any fixed $m$, we have
\begin{equation}
\lim_{\hat s \to 0} q(m,\hat s) = \lim_{\hat s \to 0} \frac{1-\left(\frac{\hat sm}{n}\right)^m}{1-\left(\frac{\hat sm}{n}\right)^{m+n}} = 1.
\end{equation}
It follows that as $\hat s \to 0$, we have that $q(m_0,\hat s)-q(m_1,\hat s)$ becomes arbitrarily small, since both terms converge to $1$. Thus, for sufficiently small $\hat s$, the left-hand side of the key inequality (\ref{inequality}) becomes sufficiently small, meaning that it is in the perceived self-interest of Member $j$ to defect. Again, it is easy to show that this result also holds for the generalized model (\ref{generalform}) of the battle-win probability, bolstering the case for the result's robustness.

A complacent member perceives the disadvantage of the AI replacement movement relative to the pro-human movement to be so substantial that the latter's win probability is very close to $1$, regardless of whether it uses the greater-unity approach (which imposes a personal cost to the future victim) or the status quo approach (which does not). Since victory is close to guaranteed in both scenarios, a complacent member would  defect to the status quo approach, which does not impose a personal cost on them at the present time.  

In other words, complacent future victims tend to unwisely refrain from uniting with current victims to oppose the shared threat, due to their underestimation of the threat. This
unwise decision to forgo unity can reduce the pro-human movement’s probability of prevailing, as illustrated by Figure~\ref{fig:perceivedthreat}.

\begin{figure}[h]
\centering
\includegraphics[width=0.85\textwidth]{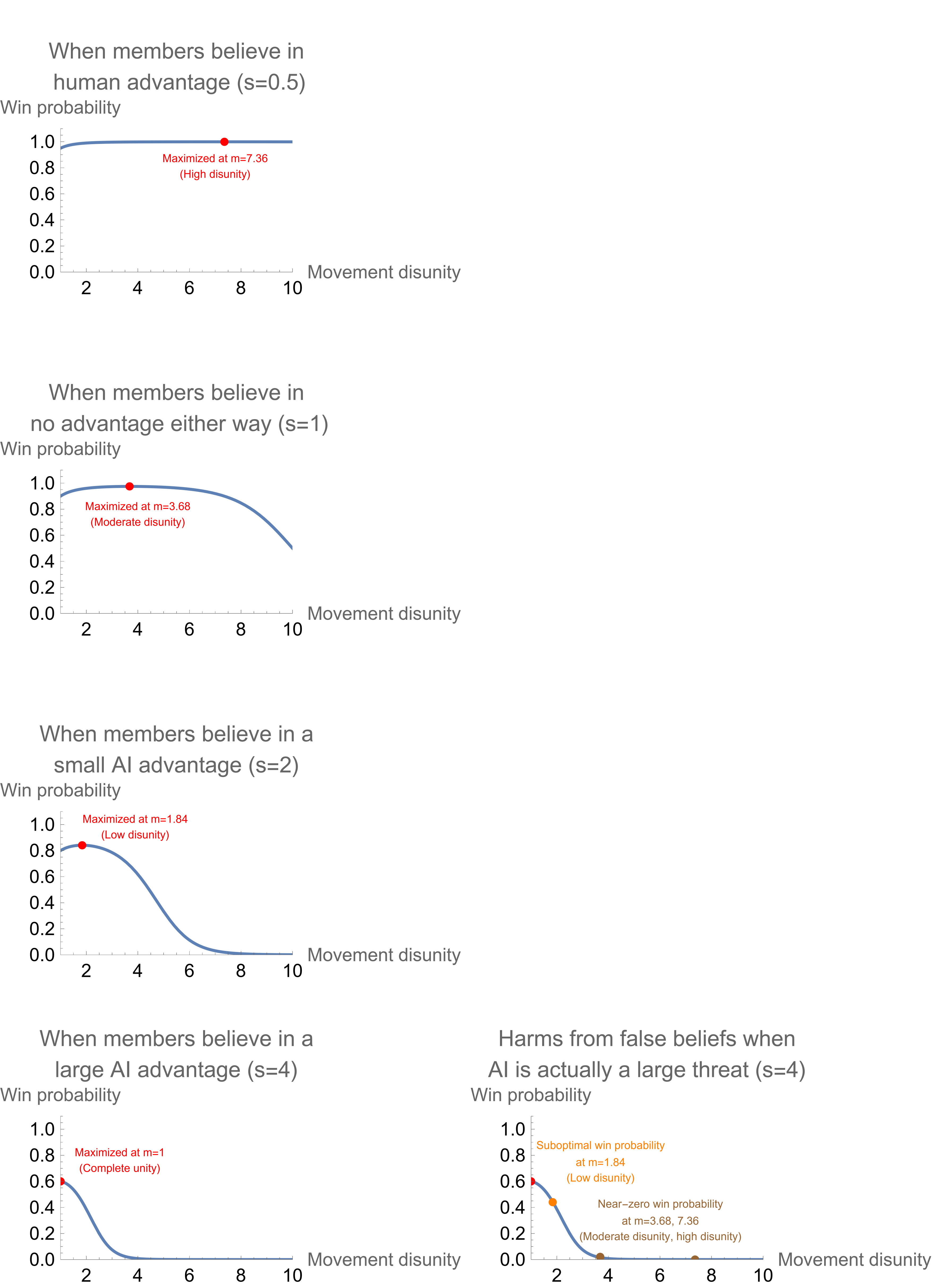}
\caption{The optimal organization structure of the pro-human movement conditional on beliefs about the AI replacement movement's strength parameter $s$, as the number of the AI replacement movement's parts $n$ is set to $n=10$. If the pro-human movement decides their organization structure based on an incorrect belief on how large the AI automation threat really is, the movement's win probability can suffer substantially.}\label{fig:perceivedthreat}
\end{figure}

\section{Corroboration by the historical-empirical record}

The strategic benefits of unity, and the strategic pitfalls of being divided and conquered, are corroborated throughout the historical-empirical record. Empirical corroboration exists for many types of conflicts between human groups, but is especially well-documented for military conflicts, as well as for public-relations conflicts between corporations and their opposition movements. Also contained in the historical-empirical record are illustrative examples of AI-industry leaders'  pronouncements, which we find to be consistent with our model's predictions on what public-relations outcomes a corporate campaign would find helpful for dividing-and-conquering its opposition movement.

\subsection{The pronouncements of AI-industry leaders are consistent with our model's predictions}

 AI-industry leaders have consistently made pronouncements that, if taken at face value by potential members of the pro-human movement, would have the effect of making them more myopic, naive, collaborationist, defeatist, and complacent. We reiterate the definitions of these five concepts below.

First, a movement is more likely to be divided-and-conquered if its members are \emph{myopic}, in that they deprioritize their future well-being relative to their present well-being. Myopia makes members less likely to pay a personal cost today in order to achieve greater movement unity and thereby, a greater probability of collectively prevailing over the AI replacement movement in the future.

Second, a movement is more likely to be divided-and-conquered if its members are \emph{naive}, in that they underestimate the harms that the adversary's takeover would inflict on them. 

Third, a movement is more likely to be divided-and-conquered if its members are \emph{collaborationist}, in that they believe they will actually personally benefit from the adversary taking over.

Fourth, a movement is more likely to be divided-and-conquered if its members are \emph{defeatist}, in that they believe the adversary's takeover is inevitable, regardless of their choice of actions.

Finally, a movement is more likely to be divided-and-conquered if its members are \emph{complacent}, in that they underestimate the degree to which society is vulnerable to  takeover by AI-driven disempowerment, until it is too late.

We present in Table~\ref{tab:quotes} some representative examples of  AI-industry leaders' pronouncements that, if taken at face value, are consistent with inflicting these five weaknesses on potential members of the pro-human movement. It is important to note that these  pronouncements by themselves do not suffice to prove an explicit intent to divide-and-conquer the pro-human movement. Regardless, they bolster the case for further investigation.

\subsection{Divide-and-conquer dynamics in corporate influence campaigns}

Many corporate influence campaigns utilize divide-and-conquer dynamics. By methods such as public relations, discreet lobbying, and the formation of front groups, corporate influence campaigns can  divert negative attention away from the relevant activities or products by fragmenting the opposition movement. \supercite{beder2002global} Below, we outline several noteworthy instances of corporate divide-and-conquer campaigns in the historical-empirical record.

\textit{Keep America Beautiful}

Arguably one of the most infamous corporate divide-and-conquer campaigns is Keep America Beautiful, an anti-littering nonprofit responsible for extremely influential ad campaigns (see Figure~\ref{fig:keepamerica}) en route to becoming the largest community-improvement organization in the United States. \supercite{KeepAmericaBeautiful2023} Despite attracting many well-meaning adherents to this day, Keep America Beautiful was in fact created by the packaging industry and its associates, such as Coca-Cola and the Dixie Cup Company. \supercite{venkatesan2021marketing} The effect of the group's anti-littering campaign was to shift the responsibility for packaging waste to individuals, rather than the corporations that produce it. By framing packaging waste as a problem of individual choice rather than that of corporations' unsustainable packaging,  the corporations effectively divided the environmentalist movement against packaging waste into those who felt the need to clean up packaging waste right now, from those who saw it as a distraction from reducing society's long-run reliance on unsustainable packaging.

\newpage
\pagebreak

\begin{table}[h]
\centering
\caption{AI-industry leaders' words that, if taken at face value, could inflict the weaknesses predicted by our model---myopia, naivety, collaborationism, defeatism, and complacency---on the pro-human movement.  }\label{tab:quotes}
\small

\vspace{9pt}

\begin{tabular}{ | m{7.2em} | m{10.5cm}| } 
\hline
\hspace{12pt} \textbf{Weakness} & \hspace{2pt} \textbf{Words of AI-industry leaders that are consistent with causing the weakness}  \\
\hline
\begin{center}
 \vspace{8pt}
 \textbf{myopia}

 \vspace{4pt}
 
  (high $r $)
  \end{center}
 & 
\vspace{1pt}

\hspace{-3pt}\textit{``I don't know if mundane is the right word, but there are concerns that already exist, about people using AI tools to do harmful things of the type that we’re already aware...}

\vspace{2pt}

\textit{...That’s going to be a pretty big set of challenges that the companies working on this are going to need to grapple with, regardless of whether there is an existential crisis as well sometime down the road.''}
\begin{flushright}
\vspace{-2pt}
- Mark Zuckerberg, CEO of Meta \supercite{fridman2023zuckerberg}
\vspace{-9pt}
\end{flushright}
\\ \hline
\begin{center}
 \vspace{8pt}
 \textbf{naivety}

  \vspace{4pt}

 (low $b/c$)
   \end{center}
 &
\vspace{1pt}

\hspace{-3pt}\textit{``Once AI systems become more intelligent than humans, ... we will \textbf{still} be the 'apex species'...} 

\vspace{2pt}

\textit{...We will design AI to be like the supersmart-but-non-dominating staff member.''}
\begin{flushright}
\vspace{-2pt}
- Yann LeCun, Chief AI Scientist of Meta \supercite{LeCun2023tweet}
\vspace{-9pt}
\end{flushright}
\\ \hline
\begin{center}
 \vspace{8pt}
\textbf{collaborationism} 

 \vspace{4pt}
 
($b<0$)
\end{center}& 
\vspace{1pt}

\hspace{-3pt}\textit{``Unless we built in safeguards, Musk argued, artificial-intelligence-systems might replace humans, making our species irrelevant or even extinct.}

\vspace{1pt}

\textit{Page pushed back. Why would it matter...if machines someday surpassed humans in intelligence, even consciousness? It would simply be the next stage of evolution.''}
\begin{flushright}
\vspace{-2pt}
-  Larry Page, Co-Founder of Google\supercite{isaacson2023inside}
\end{flushright}

\vspace{4pt}

\textit{``F*** AI safety...Me and my robot homies are gonna come to your house.''}
\begin{flushright}
\vspace{-2pt}
-  Martin Shkreli, Co-Founder of DL Software\supercite{Shkreli2023,reuters2023}
\end{flushright}

\vspace{4pt}

\textit{``Rather quickly, they would displace us from existence...It behooves us to give them every advantage, and to bow out when we can no longer contribute...}

\vspace{1pt}

\textit{...I don't think we should fear succession. I think we should not resist it. We should embrace it and prepare for it. Why would we want greater beings, greater AIs, more intelligent beings kept subservient to us?''}
\begin{flushright}
\vspace{-2pt}
- Rich Sutton, \supercite{sutton2023aisuccession} First-Ever Advisor of Google DeepMind \supercite{hassabis}

\vspace{-9pt}
\end{flushright}
\\ \hline
\begin{center}
 \vspace{8pt}
\textbf{defeatism}  

 \vspace{4pt}
 
(very high $\hat s$)
\end{center}
& 
\vspace{1pt}

\hspace{-3pt}\textit{“AI will probably most likely lead to the end of the world, but in the meantime, there'll be great companies.”}
\begin{flushright}
\vspace{-2pt}
- Sam Altman, CEO of OpenAI \supercite{DeVynck2023}
\end{flushright}

\vspace{5pt}

\textit{``AI is the new default to build all tech and we’re here for it!''}
\begin{flushright}
\vspace{-2pt}
- Clem Delangue, CEO of Hugging Face \supercite{Delangue2023}
\end{flushright}

\vspace{4pt}

\textit{“Imagine if everyone of good conscience said, ‘I don’t want to be involved in building A.I. systems at all. Then the only people who would be involved would be the people who ignored that dictum — who are just, like, ‘I’m just going to do whatever I want.’"}
\begin{flushright}
\vspace{-2pt}
- Dario Amodei, CEO of Anthropic \supercite{Amodei2023}
\end{flushright}

\vspace{4pt}

\textit{``The succession to AI is inevitable...Inevitably, eventually, it would become more important in almost all ways than ordinary humans.''}
\begin{flushright}
\vspace{-2pt}
- Rich Sutton, \supercite{sutton2023aisuccession}  First-Ever Advisor of Google DeepMind \supercite{hassabis}
\vspace{-10pt}
\end{flushright}

\\ \hline
\begin{center}
 \vspace{8pt}
\textbf{complacency} 

 \vspace{4pt}

 (very low $\hat s$) 
 \end{center}
& 
\vspace{1pt}
\hspace{-3pt}\textit{``Worrying about AI today is like worrying about overpopulation on Mars.''}
\begin{flushright}
\vspace{-2pt}
- Andrew Ng, Co-Founder of Google Brain, Former Chief Scientist of Baidu \supercite{ng2016worrying}
\end{flushright}

\vspace{4pt}

\textit{“Why debate human-level AI when we can't approach dog-level intelligence yet?”}
\begin{flushright}
\vspace{-2pt}
- Yann LeCun, Chief AI Scientist of Meta \supercite{LeCun2023}
\vspace{-10pt}
\end{flushright}
\\ \hline
\end{tabular}
\end{table}
\newpage
\pagebreak

\begin{figure}[h]
\centering

\includegraphics[width=0.7\textwidth]{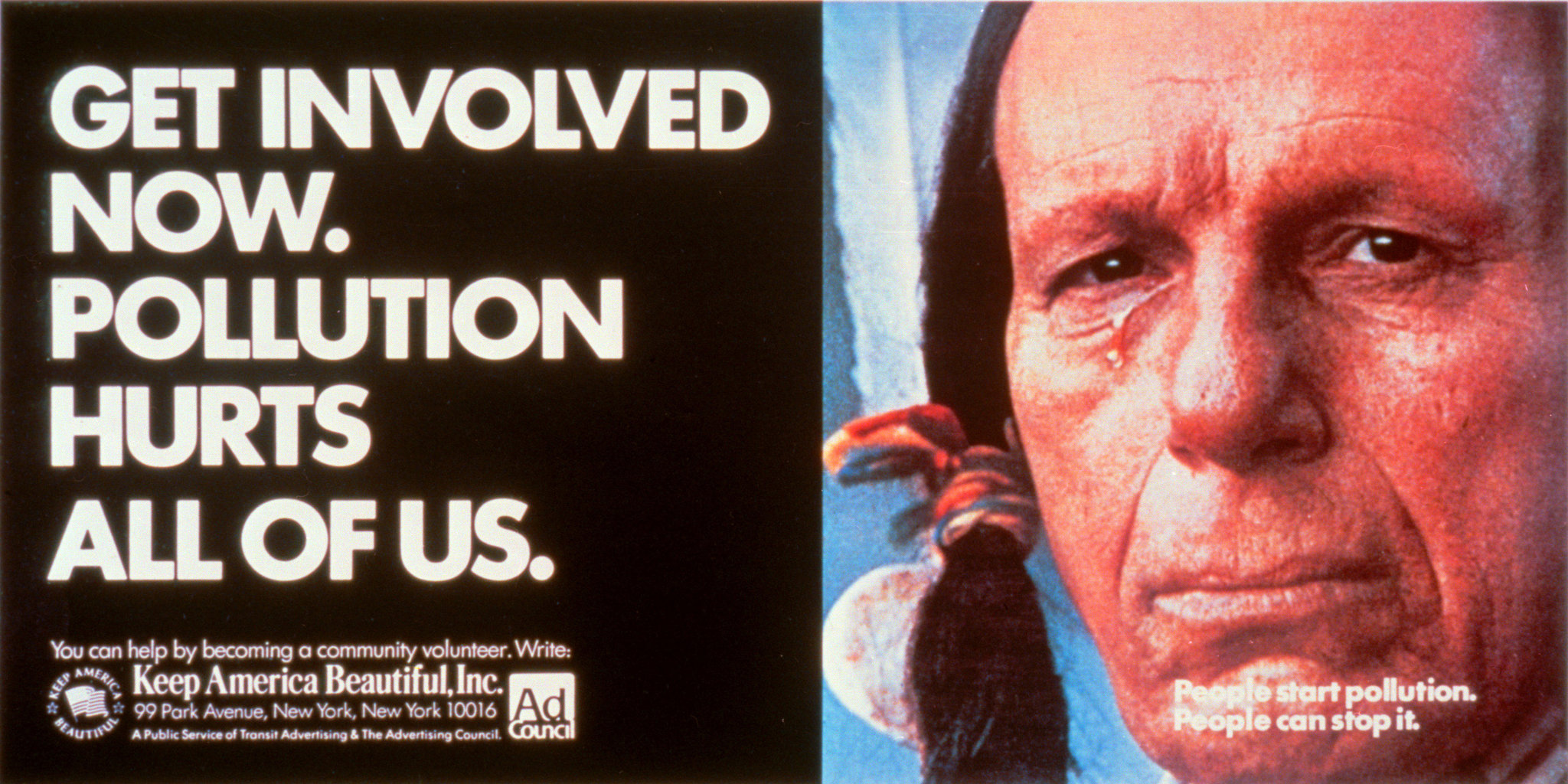}

\caption{Keep America Beautiful ran an extremely influential public service announcement, in which a Native American man sheds a tear as he surveys an environment filled with litter. \supercite{levere2013after}  Messages in the ad campaign like \textit{``Get involved now. Pollution hurts all of us.''} and \textit{``People start pollution. People can stop it.''} divided environmentalists who felt the need to clean up litter right now, from other environmentalists who saw it as a distraction from reducing society's long-run reliance on unsustainable packaging. The campaign was in fact intended to serve corporate interests.}\label{fig:keepamerica}

\end{figure}

\textit{Council for Solid Waste Solutions}

Another example of a corporate divide-and-conquer strategy is provided by the Council for Solid Waste Solutions, a special group that promoted the recycling of plastic. The group was instrumental to establishing the modern plastic recycling system. In actuality, the group was created by the Society of Plastics Industry---comprised of plastic-industry corporations like Exxon, Mobil, Dow, DuPont, Chevron, and Phillips 66---to divert attention away from the fundamentally waste-generating nature of its products. \supercite{bonta2023plastics} Indeed, to this day, well-meaning adherents are misled into believing that plastic is recyclable and therefore sustainable. This is especially well-exemplified by the universal recycling symbol, the ubiquitous triangle of arrows that the plastic industry quietly lobbied most U.S. states to mandate on all plastic, including the many variants that are non-recyclable. \supercite{sullivan2020bigoil} In actuality, less than 10\% of plastic is successfully recycled, \supercite{EPA2023} and commissioned scientists had warned plastic executives as early as 1973 that recycling plastic is ``costly'' and that sorting it is ``infeasible.'' \supercite{sullivan2020bigoil} The plastic industry successfully divided the environmentalist movement against plastic waste into those who felt the need to recycle plastic now versus those who saw it as a distraction from reducing society's long-run reliance on plastic-waste-producing products.

\subsection{Divide-and-conquer dynamics in military conflicts}

 Throughout human history, the phenomenon of divide-and-conquer---often referred to as \emph{defeat in detail} in military terminology---has been one of the most important considerations for how to wage war effectively. \supercite{erickson2003defeat} We discuss three representative examples below.

\textit{Napoleon's first Italian campaign (1796-97)}

Napoleon's unusually high rate of military victory is often attributed to his ability ``to overcome a superior enemy by dividing their forces.'' \supercite[p. 3]{durham2015command} This is especially well-demonstrated by his first Italian campaign during the War of the First Coalition. In it, Napoleon adeptly employed speed, maneuverability, and central positioning of his 38,000 troops to divide the defending forces, comprised of 52,000 Piedmontese and Austrian forces. \supercite{TheFrenchWayofWar} By defeating the former into surrendering and subsequently vanquishing the latter, Napoleon was able to achieve a complete victory over a numerically superior adversary.

\textit{The First Balkan War (1912-13)}

Divide-and-conquer dynamics were of pivotal significance in the First Balkan War. This is well-explained by the following preamble quote from Edward Erickson~\supercite[p. xvii]{erickson2003defeat}:

\begin{quote}
"The Ottoman Empire engaged in the Balkan Wars of 1912–1913 against the Balkan League (composed of
Bulgaria, Greece, Montenegro, and Serbia) and was decisively defeated in detail. In short summary, the Ottoman Empire split its field armies into groups and thereby created the conditions necessary for its enemies to achieve numerical advantage on the battlefield...

...The Ottoman armies were then defeated in succession."
\end{quote}

Considered as one of the world's great powers before the First Balkan War, the Ottomon Empire's surprising defeat led to the disastrous outcome of ceding almost all of its Balkan territories.

\textit{Operation COMPASS during World War II (1941)}

In the North African front of World War II, 150,000 Italian troops in Italy-controlled Libya threatened British-administered Egypt, erecting makeshift defenses around Sidi Barrani against a comparatively small British force of just 36,000 troops. \supercite{playfair2004mediterranean} Nevertheless, the small but ambitious British force responded with Operation COMPASS, during which they consistently outmaneuvered the comparatively low-mobility Italian forces. Despite their four-to-one numerical advantage, the Italian forces were routed into surrender by the British, with 130,000 of their troops taken prisoner. \supercite{australian} Because the Italians' defensive positions were far away from each other relative to their mobility, these defensive positions were unable to effectively support each other, enabling a British strategy of divide-and-conquer. The frantic and fragmented retreat by the Italian forces afterwards only made them further vulnerable to being defeated in detail. 
Synergistic cooperation between nearby allies constitutes an important mechanism by which superior unity can enable one side of a conflict to prevail against the other.

\section{Discussion}

Like military conflicts, humanity's conflict against AI-driven disempowerment can also be decomposed into a collection of ``battlefields'' spanning  the legal, political, and social spheres. While these ``battlefields'' are not equivalent to military battlefields, they may be analogous as far as divide-and-conquer dynamics are concerned. For instance, just as military units are generally made more effective by the cooperation of supporting military units, stakeholders of the pro-human movement may become more effective at fighting AI-driven disempowerment if they were also supported by other like-minded stakeholders. Moreover, just as military units with different skillsets can cooperation synergistically, so to can the different stakeholders of the pro-human movement, whose sum may be greater than the parts in isolation. We discuss some current and potential future examples below.

\subsection{Moral critique of AI collaborationists}

Some individuals adopt a collaborationist approach to AI. A particularly egregious example is provided by Rich Sutton, the first-ever advisor to Google DeepMind.\supercite{hassabis} In his talk at the World Artificial Intelligence Conference in Shanghai---titled \emph{AI succession} \supercite{sutton2023aisuccession}---Sutton said:

\begin{quote}
``Rather quickly, they would displace us from existence...It behooves us to give them every advantage, and to bow out when we can no longer contribute...

...I don't think we should fear succession. I think we should not resist it. We should embrace it and prepare for it. Why would we want greater beings, greater AIs, more intelligent beings kept subservient to us?''

\end{quote}

Instead of being shunned for his proposal, Sutton was given a partnership with Keen Technologies to build autonomous AI systems that outperform humans at most economically relevant capabilities. \supercite{alberta2023}

For a less egregious, but more widespread example of AI collaborationism, consider the AI specialists who work to replace other people's livelihoods en route to making large amounts of money. Annual salaries of \$900,000 USD for such AI specialists have been offered by companies ranging from OpenAI \supercite{levels_fyi_2023} to Netflix.\supercite{cutter2023_900000AIJob} 

As the shared threat of AI-driven disempowerment becomes more apparent, however, such AI collaborationists may become increasingly subject to moral critique. For a successful example of such a moral critique, consider the incident of Prosecraft, a platform constructed by developer Benji Smith for the linguistic analysis of literature. Using an algorithm to draw from a massive repository of 25,000 books, Smith employed AI to obtain analytics about the books, such as word counts and adverb usage metrics. \supercite{mather2023prosecraft} While he believed these insights would help inform authors, many of them were instead incensed upon discovering that their works had been fed into Prosecraft without their consent. The ensuing uproar by authors convinced Smith to shut Prosecraft down.

 Consistent with our model of divide-and-conquer dynamics, the AI replacement movement reacted to this by using pity for Smith to negatively portray writers' uproar about Prosecraft.  To illustrate, Yann LeCun---Meta's Chief AI Scientist---attempted to martyr Smith by calling him ``a civilian casualty of the anti-AI crusade.'' \supercite{lecun2023civilian} This is consistent with an attempt to demoralize and divide the parties who would otherwise be sympathetic to efforts against AI-driven disempowerment.

\subsection{Copyright lawsuits}

As AI output becomes increasingly capable of substituting for the output of creative workers, a burgeoning frontier of conflict has emerged pertaining to these creative workers' intellectual-property rights. This becomes especially pertinent when AI systems are trained using copyrighted works. The increasing prevalence of AI-generated content has led more and more creatives to file copyright lawsuits against the companies responsible. \supercite{vincent2023ai} 

However, these copyright lawsuits have yet to see results, and this may continue to hold in the future. To illustrate, during a hearing for the class-action lawsuit filed by Sarah Andersen et al. (against Stability AI, DeviantArt, and Midjourney), the judge overseeing the case stated that he was inclined to dismiss the near-entirety of artist plaintiffs' entire complaint, albeit with the caveat that they can re-file in the future. \supercite{quach2023judge}

One potential factor contributing to artists' unsuccessful efforts to organize successful copyright lawsuits---against AI art models trained on their work---is a dearth of capable lawyers who are willing to legally advise or represent them. Alex Champandard, an activist for artists' intellectual-property rights, searched for a lawyer by writing the following plea \supercite{champandard2023dcma}:

\begin{quote}
``I am looking for [a] lawyer with experience in filing DCMA claims, specifically with regards to tools that `circumvent technological protection measures [...] that control access to copyrighted works.'...

...I realize this is like sounding the Horn of Gondor in the hope that a Unicorn might show up! I have yet to find many lawyers on \#TeamHuman that want to see cooperative and consensually-built AI.''
\end{quote}

Activists who want to help artists file promising copyright lawsuits are currently having trouble finding willing lawyers. However, there are signs that the more advanced AI models of the future may be capable of making many lawyers economically useless as well. \supercite{tan2023chatgpt} If lawyers were to realize that AI may be an imminent threat---not just to artists, but also to themselves---they may be more willing to join forces with artists in joint resistance against AI-driven disempowerment.

\subsection{Contract negotiations and strikes for AI-ban clauses}

Strikes advocating for AI-ban clauses constitute a formidable form of resistance against the encroachment of AI systems into the workforce. The Writers Guild of America (WGA) and the Screen Actors Guild (SAG-AFTRA) have mounted a joint strike against entertainment studios, with contractual protections against AI encroachment being one of the most important points of disagreement. \supercite{lawler2023aihollywood}  In the view of WGA/SAG-AFTRA member and AI consultant Justine Bateman, one factor that had contributed to the deadlock was that ``both the studios and the streamers think they can win against encroaching generative AI video tech companies,'' \supercite{bateman2023destruction} as long as they themselves pivot to AI. This thinking, however, may be mistaken.  Bateman instead predicted that generative AI companies ``are going to eat [the studios and streamers'] lunch.'' As is consistent with our model of divide-and-conquer, entertainment studios and streamers' underestimation of the threat AI poses to their own livelihoods can cause them to deprioritize movement unity: to the point of infighting with their own writers and actors at a time of industry-wide crisis. If entertainment studios and streamers had realized that they are in the same boat as actors and writers---the boat of indefinite AI replacement---it is plausible that they might instead have united with these very same actors and writers against AI encroachment into the film industry.

After their 146-day strike---the second longest in its history---the WGA won a deal with significant protections against AI encroachment. \supercite{wilkinson2023hollywood} The WGA's hard-fought and ultimately successful contract negotiation was likely facilitated by its members' sense of urgency: their decision to act early, rather than wait-and-see. This is illustrated by Bateman's words\supercite{taylor2023hollywood} during the strike:
\begin{quote}
``This is the last time any of these unions will have leverage. We see how much has happened in the last, what has it been, two months? In three years, we won't have this kind of leverage. So if we want to get anything for the next three years, we've got to go all in.'' 
\end{quote}
 WGA's contract win also helped set the tone for SAG-AFTRA's contract win,\supercite{alexander2023actors} an illustrative example of how different stakeholders of the pro-human movement can cooperate synergistically.

\subsection{Whistleblowers}

Historically, whistleblowers have often been crucial for informing the public of organizations' concealed wrongdoing,  ethical concerns, or breaches of safety protocols.\supercite{CORDIS2017289} In the framework of our model, efforts to persuade insiders of AI companies and the AI industry overall to publicly come forward with damning evidence are well-characterized as the pro-human movement's efforts to divide and conquer the AI replacement threat. Moreover, any form of support that stakeholders of the pro-human movement can provide potential whistleblowers who are on the fence, such as contacts in the government and media, legal advice, and various forms of operational support, would be aided by the synergistic cooperation enabled by a more united movement.

\subsection{Political advocacy for AI regulation}

One can argue that all the previous battlefields ultimately lead to this arguably most important battlefield. Stakeholders can advocate for the adoption of laws and regulations that presciently guide AI development, and impose robust restrictions on its application to use cases that threaten indefinite human economic uselessness. Stakeholders may, for example, advocate for the establishment of legal, regulatory, governmental, and inter-governmental infrastructures pertaining to the use of advanced AI in society.\supercite{sepasspour2023reality} Advocates of AI deregulation, in contrast, would advocate for the continued absence of such laws and regulations. Who prevails in this debate will be largely determined by the degree to which the voters, lobbying actors, and politicians comprising the pro-human movement are united or divided---for instance, in their efforts to advocate for supportive policies---relative to those comprising the AI replacement movement. If the former triumphs over the latter, temporary wins---such as the WGA's and SAG-AFTRA's contractual protections against AI encroachment---can be elevated to permanent wins at the level of government policy.

\subsection{A race to get most people hooked on AI before they mobilize against it}

The above ``battlefields'' corresponded to activities by which the pro-human movement can take the initiative. The AI replacement movement can take the initiative in a different kind of battlefield: increasing societal dependence on AI before the public can effectively mobilize against it.
 AI significant others \supercite{titcomb2023} and other AI products can be similarly designed to hook users quickly, potentially before they fully comprehend the long-term implications of their dependence. These AI applications may one day learn about individual users' preferences at an unprecedented level of detail, continually adapting and refining their interactions to maximize engagement. This hyperpersonalization, while increasing user convenience and satisfaction in the short term, can raise significant concerns about individual autonomy, privacy, and mental health in the long term. Moreover, it can potentially lead to a society unable to resist the allure of AI.

As such, one of the critical battles will be that of speed. On one side, we have the development and propagation of addictive AI technologies, fueled by wealthy corporations with substantial resources and social-influence tools at their disposal. On the other side, we have the potentially slower, but potentially broader bulwark of public mobilization and policymaking. Whether the latter will be able to keep up with the former---just like whether a military force is able to keep up with an maneuvering adversary attempting to defeat it in detail---may be especially important.

Given that speed is of the essence, a precommitment to act on the threat only after seeing clear signs of the threat manifest \supercite{coldewey2023ethicists} may all but ensure that crucial opportunities are missed. In a pandemic, for instance, many lives can be saved by acting early: by acting well before clear signs of the threat manifest.\supercite{bonardi2020fast} It is likewise pivotal that the future victims of AI replacement are mobilized against the threat---via accurate panic---sooner rather than later: before future victims join current victims' fate in losing their incomes and leverage, thereby losing their opportunity to support current victims' efforts against the shared threat. 

\section*{Acknowledgements}

We would like to thank Emily Dardaman, Tara Rezaei Kheirkhah, Vedang Lad, Ziming Liu, and Leo Yao for their thoughtful and helpful comments. We are also grateful to Mona Xue for suggesting the example of pandemics to demonstrate the importance of  acting early, before clear signs of the threat manifest. P.S.P. is funded by the MIT Department of Physics.

\section*{Data availability statement}

Figures and calculations were done via Mathematica 13, and the corresponding data are available at the Open Science Framework database \url{https://osf.io/tfnrw/}.

\printbibliography

\pagebreak
\newpage

\begin{appendices}

\section{Proof of Proposition \ref{prop:unitygood}}\label{appendix:serious}

We prove the proposition for the most general model (\ref{generalform}) of the battle-win probability. The corresponding expression for the overall win probability of Movement 1, $q$, can be computed as follows. Consider a Markov chain with three distinguished states: the initial state $i=n$, the terminal state $i=0$, and the other terminal state $i=m+n$. Conditional on being at the initial state $i=n$, the overall win probability is given by the equation
\begin{equation}\label{gamblers}
q=p E(p_{+}, n-1,1)+p (1-E(p_{+}, n-1,1))q + (1-p)  E(p_{-}, 1,m-1) q,
\end{equation}
where 
\begin{equation}
p_{\pm}=\sigma\left[\sigma^{-1}\left(\frac{1}{1+\left(sm/n\right)^{1/R}} \right)\pm\gamma\right],
\end{equation}
and $E(x,j,k)$ denotes the overall win probability of a generalized gambler's ruin in which Player 1 has a battle-win probability of $x$, Player 1's winning state is $j$ away from the initial state, and Player 2's winning state is $k$ away from the initial state. This is because the overarching Markov chain can be split into two generalized gambler's ruin Markov chains, one of which has battle-win probability $p_{+}$ and the other of which has battle-win probability $p_{-}$.

Solving for $q$, we obtain
\begin{align}
q&=\frac{p E(p_{+}, n-1,1)}{1-p (1-E(p_{+}, n-1,1))-(1-p)  E(p_{-}, 1,m-1)}
\nonumber\\&=\frac{\left(-1+\left(e^{-\gamma} \left(\frac{s m}{n}\right)^{1\over R}\right)^m\right)\left(-1+e^{\gamma} \left(\frac{s m}{n}\right)^{1\over R}\right)}{\mathcal D_0},
\end{align}
where the denominator expression is given by
\begin{align}
\mathcal D_0=&1+e^{\gamma}\Bigg( \left(e^{-\gamma} \left(\frac{s m}{n}\right)^{1\over R}\right)^m +e^{\gamma}\left(e^{-\gamma} \left(\frac{s m}{n}\right)^{1\over R}\right)^{1+m}  
\nonumber\\&\hspace{60pt}-\left(e^{-\gamma} \left(\frac{s m}{n}\right)^{1\over R}\right)^m\left(e^{\gamma} \left(\frac{s m}{n}\right)^{1\over R}\right)^n-\left(\frac{s m}{n}\right)^{1\over R}
\Bigg)
\nonumber\\&\hspace{10pt}- \left(e^{-\gamma} \left(\frac{s m}{n}\right)^{1\over R}\right)^m+e^{\gamma}\left(e^{-\gamma} \left(\frac{s m}{n}\right)^{1\over R}\right)^{1+m}\left(\left(e^{\gamma} \left(\frac{s m}{n}\right)^{1\over R}\right)^n-1\right).
\end{align}

Taking the partial derivative with respect to $m$, we obtain
\begin{equation}
\frac{\partial q}{\partial m}=\frac{\mathcal N_1}{\mathcal D_1},
\end{equation}
with numerator expression
\begin{align}
\mathcal N_1=& \left(e^{-\gamma} \left(\frac{s m}{n}\right)^{1\over R}\right)^m\nonumber\\&\cdot \Bigg[ \frac{1}{R} \Bigg(e^{\gamma} \left(1+\left(\frac{s m}{n}\right)^{2\over R}\right)\Bigg(m\left(-1+\left(e^{\gamma} \left(\frac{s m}{n}\right)^{1\over R}\right)^n\right)
\nonumber\\& \hspace{120pt}-\left(-1+\left(e^{-\gamma} \left(\frac{s m}{n}\right)^{1\over R}\right)^m\right)\left(e^{\gamma} \left(\frac{s m}{n}\right)^{1\over R}\right)^n n\Bigg)\nonumber\\&\hspace{30pt}+e^{2\gamma}\left(\frac{s m}{n}\right)^{1\over R}\Bigg(m-m\left(e^{\gamma} \left(\frac{s m}{n}\right)^{1\over R}\right)^n \nonumber\\&\hspace{85pt}+ \left(-1+\left(e^{-\gamma} \left(\frac{s m}{n}\right)^{1\over R}\right)^m\right)\nonumber\left(1+(n-1)\left(e^{\gamma} \left(\frac{s m}{n}\right)^{1\over R}\right)^n  \right)   \Bigg)
\nonumber\\&\hspace{30pt}+\left(\frac{s m}{n}\right)^{1\over R}\Bigg(m-m\left(e^{\gamma} \left(\frac{s m}{n}\right)^{1\over R}\right)^n \nonumber\\&\hspace{85pt}+ \left(-1+\left(e^{-\gamma} \left(\frac{s m}{n}\right)^{1\over R}\right)^m\right)\nonumber\left(-1+(n+1)\left(e^{\gamma} \left(\frac{s m}{n}\right)^{1\over R}\right)^n  \right)   \Bigg)\Bigg)
\nonumber\\&\hspace{30pt}+m\left(-1+\left(e^{\gamma} \left(\frac{s m}{n}\right)^{1\over R}\right)^n\right)
\nonumber\\&\hspace{60pt}\cdot \left( e^{\gamma} \left(1+\left(\frac{s m}{n}\right)^{2\over R}\right) - (e^{2\gamma}+1)\left(\frac{s m}{n}\right)^{1\over R} \right)\log\left(e^{-\gamma}\left(\frac{s m}{n}\right)^{1\over R}\right)\Bigg]
\end{align}
and a positive denominator expression $\mathcal D_1$.

It suffices to prove that $\mathcal N_1$ is negative for all sufficiently large $s$. This is equivalent to proving that the expression in the brackets,
\begin{align}
\frac{1}{R} &\Bigg(e^{\gamma} \left(1+\left(\frac{s m}{n}\right)^{2\over R}\right)\Bigg(m\left(-1+\left(e^{\gamma} \left(\frac{s m}{n}\right)^{1\over R}\right)^n\right)
\nonumber\\& \hspace{120pt}-\left(-1+\left(e^{-\gamma} \left(\frac{s m}{n}\right)^{1\over R}\right)^m\right)\left(e^{\gamma} \left(\frac{s m}{n}\right)^{1\over R}\right)^n n\Bigg)\nonumber\\&\hspace{30pt}+e^{2\gamma}\left(\frac{s m}{n}\right)^{1\over R}\Bigg(m-m\left(e^{\gamma} \left(\frac{s m}{n}\right)^{1\over R}\right)^n \nonumber\\&\hspace{85pt}+ \left(-1+\left(e^{-\gamma} \left(\frac{s m}{n}\right)^{1\over R}\right)^m\right)\nonumber\left(1+(n-1)\left(e^{\gamma} \left(\frac{s m}{n}\right)^{1\over R}\right)^n  \right)   \Bigg)
\nonumber\\&\hspace{30pt}+\left(\frac{s m}{n}\right)^{1\over R}\Bigg(m-m\left(e^{\gamma} \left(\frac{s m}{n}\right)^{1\over R}\right)^n \nonumber\\&\hspace{85pt}+ \left(-1+\left(e^{-\gamma} \left(\frac{s m}{n}\right)^{1\over R}\right)^m\right)\nonumber\left(-1+(n+1)\left(e^{\gamma} \left(\frac{s m}{n}\right)^{1\over R}\right)^n  \right)   \Bigg)\Bigg)
\nonumber\\&\hspace{30pt}+m\left(-1+\left(e^{\gamma} \left(\frac{s m}{n}\right)^{1\over R}\right)^n\right)
\nonumber\\&\hspace{60pt}\cdot \left( e^{\gamma} \left(1+\left(\frac{s m}{n}\right)^{2\over R}\right) - (e^{2\gamma}+1)\left(\frac{s m}{n}\right)^{1\over R} \right)\log\left(e^{-\gamma}\left(\frac{s m}{n}\right)^{1\over R}\right)\label{express}
\end{align}
is negative for all sufficiently large $s$. This can be proven by noting that for all sufficiently large $s$, the leading term
\begin{equation}\label{whole}
  -  \frac{1}{R} e^{\gamma} \left(1+\left(\frac{s m}{n}\right)^{2\over R}\right)\left(-1+\left(e^{-\gamma} \left(\frac{s m}{n}\right)^{1\over R}\right)^m\right)\left(e^{\gamma} \left(\frac{s m}{n}\right)^{1\over R}\right)^n n
\end{equation}
is negative and has a magnitude that exceeds the sum of the magnitudes of all nine other terms in the expression (\ref{express}). Indeed, observe that one-sixth of the expression (\ref{whole}),
\begin{equation}\label{oneover}
  -  \frac{1}{6R} e^{\gamma} \left(1+\left(\frac{s m}{n}\right)^{2\over R}\right)\left(-1+\left(e^{-\gamma} \left(\frac{s m}{n}\right)^{1\over R}\right)^m\right)\left(e^{\gamma} \left(\frac{s m}{n}\right)^{1\over R}\right)^n n,
\end{equation}
exceeds in magnitude each of the six positive terms (defined below) for sufficiently large $s$. 

We use the following lemma.

\begin{lemma}\label{lem}
Let $f,g$ be two continuous positive functions defined on $m \ge 1$. Suppose that $f(m) \gg g(m)$, in that there exists $C>0$ and $m_0 \ge 1$ such that $Cf(m) > g(m)$ for all $m>m_0$. Then, there exists $C'>0$ such that $C'f(m) > g(m)$ for all $m \ge 1$. 
\end{lemma}

The lemma is proven in two steps. First, a constant $C''>0$ is chosen large enough so that 
\begin{equation}
C''\min_{m \in [1,m_0]} f(m) > \max_{m \in [1,m_0]} g(m),
\end{equation}
where the minimum and maximum values are well-defined due to compactness. Second, we take $C'=\max(C, C'')$, so that $C'f(m) > g(m)$ for all $m \ge 1$.

We now return to the proof of the main result.  Recall that in our context, $m, n \ge 1$ and $s,R>0$.

\subsection{First term}\label{subsec:first}

First, we check that (\ref{oneover}) is larger in magnitude  than 
\begin{equation}
 \frac{1}{R} e^{\gamma} \left(1+\left(\frac{s m}{n}\right)^{2\over R}\right)m\left(-1+\left(e^{\gamma} \left(\frac{s m}{n}\right)^{1\over R}\right)^n\right)
\end{equation}
for sufficiently large $s$. In other words, we need to show that
\begin{equation}
    \frac{1}{6} \left(-1+\left(e^{-\gamma} \left(\frac{s m}{n}\right)^{1\over R}\right)^m\right)\left(e^{\gamma} \left(\frac{s m}{n}\right)^{1\over R}\right)^n n > m\left(-1+\left(e^{\gamma} \left(\frac{s m}{n}\right)^{1\over R}\right)^n\right)
\end{equation}
for sufficiently large $s$. It suffices to show that 
\begin{equation}\label{1inequality}
    \frac{1}{6} \left(-1+\left(e^{-\gamma} \left(\frac{s m}{n}\right)^{1\over R}\right)^m\right) n > m
\end{equation}
for sufficiently large $s$. This can be rearranged to the inequality 
\begin{equation}
    \left(e^{-\gamma} \left(\frac{s m}{n}\right)^{1\over R}\right)^m n > \frac{6m}{n}+1.
\end{equation}
Since $m \ge 1$, it suffices to show that 
\begin{equation}
   s^{1 \over R} \left(e^{-\gamma} \left(\frac{m}{n}\right)^{1\over R}\right)^m n > \frac{6m}{n}+1.
\end{equation}
is true for sufficiently large $s$. By Lemma~\ref{lem}, there exists $C>0$ such that 
\begin{equation}
   C \left(e^{-\gamma} \left(\frac{m}{n}\right)^{1\over R}\right)^m n > \frac{6m}{n}+1.
\end{equation}
for all $m \ge 1$. As long as $s > C^R$, our desired inequality is true.

\subsection{Second term}\label{subsec:second}

Second, we check that (\ref{oneover}) is larger in magnitude than 
\begin{equation}
 \frac{1}{R}e^{2\gamma}\left(\frac{s m}{n}\right)^{1\over R}m
\end{equation}
for sufficiently large $s$. Equivalently, we need to show that for sufficiently large $s$, we have
\begin{equation}
\frac{1}{6}e^{-\gamma} \left(\frac{s m}{n}\right)^{1\over R}\left(-1+\left(e^{-\gamma} \left(\frac{s m}{n}\right)^{1\over R}\right)^m\right)\left(e^{\gamma} \left(\frac{s m}{n}\right)^{1\over R}\right)^n n > m
\end{equation}
for all $m \ge 1$.  Without loss of generality, assume that $s$ is large enough that 
\begin{equation}
e^{-\gamma} \left(\frac{s m}{n}\right)^{1\over R}\left(e^{\gamma} \left(\frac{s m}{n}\right)^{1\over R}\right)^n n >1
\end{equation}
for all $m\ge 1$. Then, it suffices to show that, possibly by taking $s$ larger, the inequality (\ref{1inequality}) is true for all $m\ge 1$. We have already shown this in Subsection~\ref{subsec:first}.

\subsection{Third term}

Third, we check that (\ref{oneover}) is larger in magnitude  than 
\begin{equation}
\frac{1}{R} e^{2\gamma}\left(\frac{s m}{n}\right)^{1\over R} \left(-1+\left(e^{-\gamma} \left(\frac{s m}{n}\right)^{1\over R}\right)^m\right)\left(1+(n-1)\left(e^{\gamma} \left(\frac{s m}{n}\right)^{1\over R}\right)^n  \right) 
\end{equation}
Equivalently, we need to show that for sufficiently large $s$, we have
\begin{equation}
\frac{1}{6}e^{-\gamma} \left(\left(\frac{s m}{n}\right)^{-1\over R}+\left(\frac{s m}{n}\right)^{1\over R}\right)\left(e^{\gamma} \left(\frac{s m}{n}\right)^{1\over R}\right)^n n  > \left(1+(n-1)\left(e^{\gamma} \left(\frac{s m}{n}\right)^{1\over R}\right)^n  \right).
\end{equation}
Without loss of generality, we can assume that $s$ is large enough that 
\begin{equation}
    \frac{1}{6}e^{-\gamma} \left(\left(\frac{s m}{n}\right)^{-1\over R}+\left(\frac{s m}{n}\right)^{1\over R}\right) > 1
\end{equation}
for all $m \ge 1$. Then, it suffices to show that, possibly by taking $s$ larger, the inequality 
\begin{equation}
\left(e^{\gamma} \left(\frac{s m}{n}\right)^{1\over R}\right)^n n  > \left(1+(n-1)\left(e^{\gamma} \left(\frac{s m}{n}\right)^{1\over R}\right)^n  \right)
\end{equation}
holds for all $m\ge 1$, which simplifies to 
\begin{equation}
\left(e^{\gamma} \left(\frac{s m}{n}\right)^{1\over R}\right)^n   > 1.
\end{equation}
This inequality is clearly true for $s$ sufficiently large.

\subsection{Fourth term}

Fourth, we check that (\ref{oneover}) is larger in magnitude than \begin{equation}
 \frac{1}{R}\left(\frac{s m}{n}\right)^{1\over R}m
\end{equation}
for sufficiently large $s$. This follows from the result shown in Subsection~\ref{subsec:second}

\subsection{Fifth term}
Fifth, we check that (\ref{oneover}) is larger in magnitude than 
\begin{equation}
 \frac{1}{R}\left(\frac{s m}{n}\right)^{1\over R} \left(-1+(n+1)\left(e^{\gamma} \left(\frac{s m}{n}\right)^{1\over R}\right)^n  \right).
\end{equation}

Equivalently, we need to show that for sufficiently large $s$, we have
\begin{equation}\label{thein}
\frac{1}{6}e^{\gamma} \left(\left(\frac{s m}{n}\right)^{-1\over R}+\left(\frac{s m}{n}\right)^{1\over R}\right) > \left(\frac{1}{\left(e^{\gamma} \left(\frac{s m}{n}\right)^{1\over R}\right)^n n }+\frac{n+1}{n}  \right).
\end{equation}
As $s \to \infty$, the right-hand side is upper-bounded by a constant, whereas the left-hand side grows to infinity. Thus, for sufficiently large $s$, the inequality  (\ref{thein}) holds.

\subsection{Sixth term}
Finally, we check that (\ref{oneover}) is larger in magnitude than the positive term in
\begin{align}
 \frac{1}{R}m&\left(-1+\left(e^{\gamma} \left(\frac{s m}{n}\right)^{1\over R}\right)^n\right)
\nonumber\\&\hspace{20pt}\cdot \left( e^{\gamma} \left(1+\left(\frac{s m}{n}\right)^{2\over R}\right) - (e^{2\gamma}+1)\left(\frac{s m}{n}\right)^{1\over R} \right)\log\left(e^{-\gamma}\left(\frac{s m}{n}\right)^{1\over R}\right).\label{sixthterm}
\end{align}
Without loss of generality, we can assume that $s$ is large enough that the factors $\left(-1+\left(e^{\gamma} \left(\frac{s m}{n}\right)^{1\over R}\right)^n\right)$ and $\log\left(e^{-\gamma}\left(\frac{s m}{n}\right)^{1\over R}\right)$ are positive for all $m\ge 1$. This means that the positive term in the expression (\ref{sixthterm}) is
\begin{equation}
    \frac{1}{R}m\left(-1+\left(e^{\gamma} \left(\frac{s m}{n}\right)^{1\over R}\right)^n\right)e^{\gamma} \left(1+\left(\frac{s m}{n}\right)^{2\over R}\right)\log\left(e^{-\gamma}\left(\frac{s m}{n}\right)^{1\over R}\right) 
\end{equation}
We need to show that for all sufficiently large $s$,
\begin{equation}
 \frac{n}{6m}  \left(-1+\left(e^{-\gamma} \left(\frac{s m}{n}\right)^{1\over R}\right)^m\right) > \left(-\frac{1}{\left(e^{\gamma} \left(\frac{s m}{n}\right)^{1\over R}\right)^{n}}+1\right)\log\left(e^{-\gamma}\left(\frac{s m}{n}\right)^{1\over R}\right)
\end{equation}
 for all $m\ge 1$. It suffices to show that  for all sufficiently large $s$,
 \begin{equation}
 \frac{n}{6m} \left(-1+\left(e^{-\gamma} \left(\frac{s m}{n}\right)^{1\over R}\right)^m\right)  >\log\left(e^{-\gamma}\left(\frac{m}{n}\right)^{1\over R} \right) + \frac{1}{R} \log s
\end{equation}
 for all $m\ge 1$. This can be shown by proving that for sufficiently large $s$, the two inequalities 
  \begin{equation}\label{firstofinequality}
 \frac{n}{12m} \left(-1+\left(e^{-\gamma} \left(\frac{s m}{n}\right)^{1\over R}\right)^m\right)  >\log\left(e^{-\gamma}\left(\frac{m}{n}\right)^{1\over R} \right)
\end{equation}
and 
 \begin{equation}\label{secondofinequality}
 \frac{n}{12m} \left(-1+\left(e^{-\gamma} \left(\frac{s m}{n}\right)^{1\over R}\right)^m\right)  > \frac{1}{R} \log s
\end{equation}
are true for all $m \ge 1$. To show the first inequality (\ref{firstofinequality}), rearrange it to the form
 \begin{equation}
 \frac{n}{12m} \left(e^{-\gamma} \left(\frac{s m}{n}\right)^{1\over R}\right)^m  >\log\left(e^{-\gamma}\left(\frac{m}{n}\right)^{1\over R} \right) + \frac{n}{12m}
\end{equation}
and apply Lemma~\ref{lem} in a similar manner as in Subsection~\ref{subsec:first}. To show the second inequality (\ref{secondofinequality}), rearrange it to the form
 \begin{equation}
 \frac{n}{12m} \left(e^{-\gamma} \left(\frac{ m}{n}\right)^{1\over R}\right)^m  s^{m \over R} >\frac{1}{R} \log s+ \frac{n}{12m},
\end{equation}
observe that this would follow from showing that 
 \begin{equation}
\min_{m \ge 1}\left[ \frac{n}{12m} \left(e^{-\gamma} \left(\frac{ m}{n}\right)^{1\over R}\right)^m \right]  s^{1 \over R} >\frac{1}{R} \log s+ \frac{n}{12}
\end{equation}
for all sufficiently large $s$, and noting that this is true because polynomial growth is faster than logarithmic growth. 

This concludes the proof of the monotonic decrease.

\subsection{Proof of convergence to zero}

We see that the limit of $q$ as $m \to\infty$ is zero because the numerator expression 
\begin{equation}
\left(-1+\left(e^{-\gamma} \left(\frac{s m}{n}\right)^{1\over R}\right)^m\right)\left(-1+e^{\gamma} \left(\frac{s m}{n}\right)^{1\over R}\right)
\end{equation}
is negligible compared to the leading term of the denominator expression $\mathcal D_0$, given by
\begin{equation}\label{leadingt}
e^{\gamma}\left(e^{-\gamma} \left(\frac{s m}{n}\right)^{1\over R}\right)^{1+m}\left(\left(e^{\gamma} \left(\frac{s m}{n}\right)^{1\over R}\right)^n-1\right).
\end{equation}
\end{appendices}

\end{document}